\documentclass[11pt,a4paper]{article}

\usepackage{authblk}
\usepackage{mathrsfs}

\usepackage{graphicx}
\usepackage{amsmath}
\usepackage{amsthm}
\usepackage{amssymb}
\usepackage{amsfonts}
\usepackage{epsfig}
\usepackage{hyperref}
\usepackage{color}
\usepackage[title,titletoc]{appendix}
\def\arg {\mathop{\rm arg}\nolimits}
\def\Re {\mathop{\rm Re}\nolimits}
\def\Res {\mathop{\rm Res}\nolimits}
\def\Ai {{\rm Ai}}

\newtheorem{pro}{Proposition}

\newtheorem{thm}{Theorem}[section]

\theoremstyle{remark}
\newtheorem{rem}[thm]{Remark}
\newtheorem{rhp}[thm]{RH problem}

\numberwithin{equation}{section}

\begin{document}

\title{Asymptotics of the determinant of the modified Bessel functions and the second Painlev\'e equation }

\author[a]{Yu Chen}
\author[b,\thanks{Author to whom any correspondence should be addressed, E-mail: xushx3@mail.sysu.edu.cn}] {Shuai-Xia Xu}
\author[a]{Yu-Qiu Zhao}
\affil[a]{ Department of Mathematics, Sun Yat-sen University, Guangzhou, 510275,
China}
\affil[b]{Institut Franco-Chinois de l'Energie Nucl\'{e}aire, Sun Yat-sen University, Guangzhou, 510275, China}

\date{}

\maketitle
\begin{abstract}
In the paper, we consider the extended  Gross-Witten-Wadia unitary matrix model by introducing a logarithmic term in the potential. The partition function of the model can be expressed equivalently in terms of the Toeplitz determinant with the $(i,j)$-entry being the modified Bessel functions of order $i-j-\nu$,  $\nu\in\mathbb{C}$.
When the degree $n$ is finite, we show that the Toeplitz determinant is described by the isomonodromy $\tau$-function of the Painlev\'{e} III equation. As a double scaling limit, 
we establish an asymptotic  approximation of the logarithmic derivative of the Toeplitz determinant, expressed in terms of  the Hastings-McLeod solution of the inhomogeneous Painlev\'{e} II equation with parameter $\nu+\frac{1}{2}$. The asymptotics of  the leading coefficient and  recurrence coefficient of the associated orthogonal polynomials are also derived. We obtain the results  by applying the Deift-Zhou nonlinear steepest descent method to the Riemann-Hilbert problem for orthogonal polynomials on the Hankel loop.  The main concern here is the construction of a local parametrix at the critical point $z=-1$, where the $\psi$-function of the Jimbo-Miwa Lax pair for the inhomogeneous Painlev\'{e} II equation  is involved.
\end{abstract}

\tableofcontents

\noindent

\section{Introduction}

Consider the  following partition function of the matrix model
\begin{equation}\label{eq:URMT}
Z_{n,\nu}(t)=\frac{1}{n!}\left(\prod_{k=1}^{n}\int_{\Gamma}\frac{ds_{k}}{2\pi is_{k}}\right)\Delta(\mathbf{s})\Delta(\mathbf{s}^{-1})e^{\frac{t}{2}\sum\limits_{k=1}^{n} (s_{k}+s_{k}^{-1})+\nu \log s_k},
\end{equation}
where $\nu\in\mathbb{C}$, $t>0$,  $\Delta(\mathbf{s})=\prod_{i<j}(s_{i}-s_{j})$ is the Vandermonde determinant and  the branch of $\log s_k$ is chosen such that $\arg s_k\in(-\pi,\pi)$.  Here, the integral path $\Gamma$ denotes the Hankel  loop as shown in Figure \ref{hankelloop}, which starts at $-\infty$, encircles the origin once in the positive direction and returns to  $-\infty$.
\begin{figure}[h]
  \centering
  \includegraphics[width=6.5cm,height=4cm]{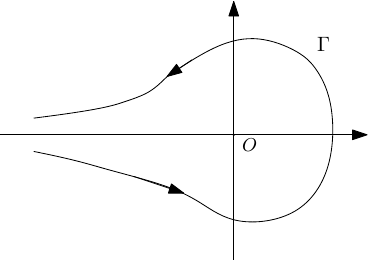}\\
  \caption{The Hankel loop  $\Gamma$}\label{hankelloop}
\end{figure}
The matrix model has been obtained recently in the  studies of  irregular conformal blocks  \cite{IOY,IOY19-2, IOY10}. The matrix model \eqref{eq:URMT} has also been considered in the studies of  random matrix average over unitary group with Haar measure earlier in \cite{FW-2003,FW}.
When $\nu$ is an integer, the contour is reduced to the positively oriented unit circle around the origin,
and the model \eqref{eq:URMT} becomes the unitary  matrix model extending  the classical Gross-Witten-Wadia unitary matrix model \cite{GW,W} by introducing the extra logarithmic term in the potential \cite{FW}.

The partition function \eqref{eq:URMT} can be expressed equivalently in terms of the Toeplitz determinant and in terms of certain   orthogonal polynomials. To this aim, we introduce the systems of  monic orthogonal polynomials   $\{\pi_n\}_{n\in\mathbb{N}}$ and $ \{\tilde{\pi}_n\}_{n\in\mathbb{N}}$ defined by the orthogonality on the Hankel loop depicted in Figure \ref{hankelloop}:
 \begin{equation}\label{eq:Ortho}
\int_{\Gamma} \pi_n(s) \tilde{\pi}_m(s^{-1}) w(s) \frac{ds}{2\pi is}=h_n\delta_{n,m}, \end{equation}
where $h_n$ is a constant, and
 \begin{equation}\label{eq: weight}
w(z)=e^{\frac{t}{2}(z+\frac{1}{z})+\nu \log z}, \quad \nu\in \mathbb{C}, 
\end{equation}
with the branch of $\log z$ chosen such that $\arg z\in(-\pi,\pi)$.
For integer $\nu$, the orthogonality \eqref{eq:Ortho} is equivalent to a non-Hermitian orthogonality on the unit circle oriented counterclockwise, which was studied in \cite{FW}.
We denote
$$m_k=m_k(t)=\int_{\Gamma} s^{k} w(s) \frac{ds}{2\pi is}, \quad k\in\mathbb{Z}.$$
The Toeplitz determinant associated with the weight function \eqref{eq: weight} can be defined by
\begin{equation}\label{eq: Top}
D_{n,\nu}(t)=\det\left(m_{j-i}\right)_{i,j=0}^{n-1}. \end{equation}
Then, the partition function \eqref{eq:URMT} can be expressed  in terms of the  Toeplitz determinant
\begin{equation}\label{eq: TopZ}
Z_{n,\nu}(t)= D_{n,\nu}(t).
 \end{equation}
From the integral representation of the $I$ Bessel function \cite{AS,Olver}, we have
\begin{equation}\label{eq: Bes}
m_k(t)=I_{-k-\nu}(t), \end{equation}
and
\begin{equation}\label{TopBes}
Z_{n,\nu}(t)= D_{n,\nu}(t)=\det \left(I_{i-j-\nu}(t)\right)_{i,j=0}^{n-1}.
\end{equation}
Here $I_{-\nu}(t)$ denotes the $I$ Bessel function of order $-\nu\in\mathbb{C}$:
$$I_{-\nu}(t)=\int_{\Gamma}  w(s) \frac{ds}{2\pi is},$$
where $w(z)$ is defined in \eqref{eq: weight} and $\Gamma$ is the Hankel loop depicted in Figure \ref{hankelloop}; see \cite[(9.6.20)]{AS}.

Suppose the determinant $D_{n,\nu}(t)$ ($D_{n,\nu}$, for short) does not vanish,
then the orthogonal polynomials can  be constructed explicitly as follows:
\begin{equation}\label{eq: OPexpre}
\pi_n(z)=\frac{1}{D_{n,\nu}}\left| \begin{array}{cccc}m_0 & m_{1} & \cdots & m_{n} \\m_{-1} & m_0 & \cdots & m_{-1+n} \\ \vdots &  \vdots &  \vdots &  \vdots \\
m_{-n+1} & m_{-n+2} & \cdots & m_{1} \\1 & z & \cdots & z^{n}\end{array} \right|,
 \end{equation}
 and
 \begin{equation}\label{eq: tildeOPexpre}
 \tilde{\pi}_n(z)=\frac{1}{D_{n,\nu}}\left| \begin{array}{cccc}m_0 & m_{-1} & \cdots & m_{-n} \\m_1 & m_0 & \cdots & m_{1-n} \\ \vdots &  \vdots &  \vdots &  \vdots \\
m_{n-1} & m_{n-2} & \cdots & m_{-1} \\1 & z& \cdots & z^n\end{array} \right|.
\end{equation}
The equations \eqref{eq:Ortho} and \eqref{eq: OPexpre} then imply  the relation
\begin{equation}\label{eq: HTop}
h_n=\frac{D_{n+1,\nu}}{D_{n,\nu}}. \end{equation}
The monic orthogonal polynomials $\pi_{n}(z)$ and $\widetilde{\pi}_{n}(z)$ satisfy the recurrence relations as follows:
\begin{equation}\label{three-term1}
\pi_{n+1}(z)=z\pi_{n}(z)+\pi_{n+1}(0)\,\widetilde{\pi}^{*}_n(z), \quad \widetilde{\pi}^{*}_{n+1}(z)=\widetilde{\pi}^{*}_n(z)+\widetilde{\pi}_{n+1}(0)z\pi_{n}(z)
\end{equation}
where  $\widetilde{\pi}^*_{n}(z)=z^{n}\widetilde{\pi}_{n}(z^{-1})$ is the reversed polynomial associated with $\widetilde{\pi}_{n}(z)$.

The Toeplitz determinant \eqref{TopBes} has important applications in random matrix theory.
According to  Gessel's formula \cite{G}, the Toeplitz determinant \eqref{TopBes} with the parameter $\nu=0$ can be used to represent the distribution of the length of the longest increasing subsequence of random permutations of the numbers $1,2,\cdots, n$.
In the seminal work   \cite{BDJ},  Baik, Deift, and Johansson show that as $n\to\infty$ the scaled distribution  converges to the Tracy-Widom distribution
\begin{equation}\label{eq:TW}
F(s)=\exp\left(-\int_s^{+\infty}(x-s)u_{HM}(x)^2dx\right),
\end{equation}
where $u_{HM}(x)$ is the Hastings-McLeod solution of the homogeneous Painlev\'{e} II equation
\begin{equation}\label{eq:HPainleve2}
u_{HM}''=2u_{HM}^{3}+xu_{HM},  \qquad u_{HM}(x)\sim \Ai(x), \quad x\to\infty.
\end{equation}
They obtain the above celebrated result by exploring  the asymptotics of the Toeplitz determinant, which is shown to be described by
the Hastings-McLeod solution of the homogeneous Painlev\'{e} II equation in a certain double scaling limit  by using the Riemann-Hilbert method.

For general parameter $\nu$, the  Toeplitz determinant with entries given in terms of the $I$ Bessel functions \eqref{TopBes} has been considered in the studies of  random matrix average over unitary groups  \cite{F, FW-2002,FW-2003,FW}. The  Toeplitz determinant is identified therein as the $\tau$-function of the Painlev\'e III equation by using
Okamoto's development of the theory of Painlev\'e  equations \cite{OK}. Recently, the determinant with the entries  expressed in terms of the $K$ Bessel functions, instead of $I$ Bessel functions,  appears in the studies of  linear statistic of the inverse eigenvalues of  Laguerre unitary ensemble and   is related to the Painlev\'e III equation both for finite dimension  $n$ and in  the large-$n$ limit \cite{DFX}; see also \cite{CI, XDZ,XDZ2015}.

Recently, the matrix model \eqref{eq:URMT} is investigated in \cite{IOY,IOY19-2,IOY10}, motivated by the applications in the studies of  irregular conformal blocks.  For general integer parameter $\nu$, they first show that for fixed $n$ the partition function \eqref{eq:URMT} is  identified with the $\tau$-function of the Painlev\'e III equation.  Then, by
conducting appropriate double scaling analysis of the associated discrete Painlev\'e system, they formally derive the inhomogeneous Painlev\'{e} II equation for a certain  scaling function as $n$ tends to infinity.
It is worth  mentioning  that
the solution of the inhomogeneous Painlev\'{e} II equation has also arisen in the studies of multi-critical  Hermitian random matrix ensembles when the spectral singularity  collides with the interior singularity where the global eigenvalue density function vanishes quadratically; see \cite{CKV,CK}.

In recent years, there has been great interest  in  studying the  asymptotic behavior of  Toeplitz determinants.
The asymptotics   of Toeplitz determinants associated with general weight functions on the unit circle perturbed by several  Fisher-Hartwig singularities have been derived in \cite{DIK2,DIK}.  When the location of the Fisher-Hartwig singularities varies with the size of the determinants, the  transition asymptotics of the Toeplitz determinants are established in several different situations, such as the  emergence of a Fisher-Hartwig singularity \cite{CIK}, the merging of two Fisher-Hartwig singularities \cite{CK2015}, the emergence of a gap on the unit circle \cite{CC} and  the closing of a gap to which 
 a  Fisher-Hartwig singularity belongs 
 \cite{XZ}.

In the present paper, we are concerned with the partition function \eqref{eq:URMT} or equivalently the Toeplitz determinant  \eqref{TopBes}  with  entries given in terms of $I$ Bessel functions and general parameter $\nu\in\mathbb{C}$ therein.
When the degree $n$ is finite, we show that the Toeplitz determinant is described by the isomonodromy $\tau$-function of the Painlev\'{e} III equation. In the double scaling limit as the degree $n$ tends to infinity and $\frac{t}{n}\rightarrow1$, we establish an asymptotic  approximation of the Toeplitz determinant, expressed in terms of  the Hasting-McLeod solution of the inhomogeneous Painlev\'{e} II equation.  The asymptotics of  the leading coefficient and  recurrence coefficient of the associated orthogonal polynomials are also derived.
We obtain the results by applying the Deift-Zhou nonlinear steepest descent method developed in \cite{D1,DKMVZ1,DKMVZ2,DZ1}
to the  Riemann-Hilbert problem for orthogonal polynomials defined by  the orthogonality on the Hankel loop \eqref{eq:Ortho}.
The main focus  is the construction of a local parametrix at the critical point $z=-1$, where the $\psi$-function of the Jimbo-Miwa Lax pair for the inhomogeneous Painlev\'{e} II equation  is involved.

\subsection{Statement of results}
Our first result shows that, for fixed degree $n$, the ratio of the contiguous recurrence coefficients
satisfies the Painlev\'{e} III equation, and the Toeplitz determinant
equals to the isomonodromy $\tau$-function of the Painlev\'{e} III equation up to a constant factor.
\begin{thm}\label{thm1}
Let $\nu\in\mathbb{C}$ and $\pi_{n}$ be the $n$-th monic orthogonal polynomials defined by \eqref{eq:Ortho}, and
\begin{equation}\label{recu-coeff}
a_{n}\left(\frac{t}{2}\right)=-\frac{\pi_{n+1}(0)}{\pi_{n}(0)}.
\end{equation}
Then $a_{n}(t)$ satisfies the Painlev\'{e} III equation
\begin{equation}\label{PIIIequ}
a''_{n}=\frac{(a'_{n})^{2}}{a_{n}}-\frac{a'_{n}}{t}+\frac{4}{t}
\left((\nu-n)a_{n}^{2}+n+\nu+1\right )+4a_{n}^{3}-\frac{4}{a_{n}}.
\end{equation}
Moreover,  we have
\begin{equation}\label{fixed-dn}
D_{n,\nu}(t)=c~e^{\frac{1}{8}t^2}\tau \left(\frac{t}{2}\right),
\end{equation}
where $c$ is a constant and $\tau(t)$  is the Jimbo-Miwa-Ueno isomonodromy $\tau$-function of the  Painlev\'{e} III equation \eqref{PIIIequ}.
\end{thm}

\begin{rem}The identification of the Toeplitz determinant with the $\tau$-function of the Painlev\'{e} III equation \eqref{fixed-dn} is derived earlier in
 \cite[Proposition 2]{FW-2003}
 and \cite[Proposition 3.2]{FW} by using Okamoto's development of the theory of Painlev\'e  equations.
 According to Okamoto's theory of Painlev\'{e} III equation, there exists a sequence of $\tau$-functions expressed in terms of the Toeplitz determinant of a certain linear combination of  the modified  Bessel functions; see \cite{DFX, FW-2003,FW} and the references therein.
Setting $\sigma(t)=t\frac{d}{dt}\log\tau(t)$, then $\sigma(t)$
 satisfies the  $\sigma$-form of Painlev\'e III equation \cite[(C.29)]{JM}:
 \begin{equation}\label{eq:sigmaform}
 \left(t\sigma''-\sigma'\right)^2=4(2\sigma-t\sigma')((\sigma')^2-4t^2)+2(\theta_0^2+\theta_{\infty}^2)((\sigma')^2+4t^2)-16\theta_0\theta_{\infty} t\sigma',
\end{equation}
with $\theta_0=\nu-n$ and $\theta_{\infty}=-(n+\nu)$. Therefore, the logarithmic derivative of $D_{n,\nu}(t)$ given in \eqref{TopBes} provides a classical solution of the differential equation \eqref{eq:sigmaform}.
  \end{rem}

To state our result on the asymptotics of  the Toeplitz determinant $D_{n,\nu}(t)$ and the recurrence coefficient of the orthogonal polynomials,
we need a family of solutions of the general Painlev\'e II equation
\begin{equation}\label{Painleve2}
u''=2u^{3}+xu +\left(\nu+\frac{1}{2}\right),  \qquad \nu\in \mathbb{C}, 
\end{equation}
and the Hamiltonian associated with the solutions
\begin{equation}\label{eq:H}
H(x;\nu)=\frac{1}{2}z(z+x)+zu^2-\nu u, 
\end{equation}
with
\begin{equation}\label{eq:z}
z=u'-u^2-\frac{x}{2}.\end{equation}
The asymptotic behaviors of this family of solutions and the Hamiltonians are stated in the following theorem.

\begin{thm}\label{thm: PIIasy}There exist one parameter family of solutions to the general Painlev\'e II equation \eqref{Painleve2} with the parameter $ \nu\in \mathbb{C}$ such that
\begin{equation}\label{eq:uasyinfty}
u(x;\nu)=-\frac{\nu+\frac{1}{2}}{x}+O\left(x^{-\frac{5}{2}}\right),  \quad x\to+\infty\end{equation}
and
\begin{equation}\label{eq:uasyneginfty}
u(x;\nu)=\left\{\begin{aligned}
&\sqrt{-\frac{x}{2}}+\frac{\nu+\frac{1}{2}}{2x}+O\left((-x)^{-\frac{5}{2}}\right),\quad &\nu \notin\mathbb{N}, \quad &x \rightarrow-\infty,\\
&-\sqrt{-\frac{x}{2}}+\frac{\nu+\frac{1}{2}}{2x}+O\left((-x)^{-\frac{5}{2}}\right),\quad &\nu\in\mathbb{N}, \quad &x \rightarrow-\infty.\\ \end{aligned}
\right.
\end{equation}
Moreover, the  Hamiltonians associated with these solutions satisfy the asymptotic behaviors
\begin{equation}\label{eq:Hasyinfty}
H(x;\nu)=-\frac{1}{8}x^{2}+\frac{\nu^2-\frac{1}{4}}{2x}+O\left(x^{-3}\right), \quad x\to+\infty,\end{equation}
and
\begin{equation}\label{eq:Hasyneginfty}
H(x;\nu)=\left\{\begin{aligned}
&-\nu \sqrt{-\frac{x}{2}}+O\left((-x)^{-1}\right),\quad &\nu \notin\mathbb{N}, \quad &x \rightarrow-\infty,\\
&\nu \sqrt{-\frac{x}{2}}+O\left((-x)^{-1}\right),\quad &\nu \in\mathbb{N}, \quad &x \rightarrow-\infty.\\ \end{aligned}
\right.
\end{equation}
\end{thm}

\begin{rem}\label{rem:poles}
The solution $u(x;\nu)$ determined by the asymptotic behaviors \eqref{eq:uasyinfty} and \eqref{eq:uasyneginfty}
is known as the Hastings-McLeod solution of the general Painlev\'e II equation \eqref{Painleve2}; see \cite[Chapter 11]{FIKN,DZ2}.
The  Hastings-McLeod solution is meromorphic in the complex plane.
For $\nu<0$, it is shown in \cite{CKV} that the Hastings-McLeod solution $u(x;\nu)$ is pole free on the real axis.
By the uniqueness of solution,  it is seen from \eqref{eq:uasyinfty} and \eqref{eq:uasyneginfty} that for $\nu\in \mathbb{N}$
\begin{equation}
u(x;\nu)=-u(x;-\nu-1).
\end{equation}
Therefore, the solution $u(x;\nu)$, $\nu\in \mathbb{N}$ is also  pole free on the real axis.
\end{rem}

Now, we consider the double scaling limit when $n\rightarrow\infty$ and $\tau=\frac{t}{n}\rightarrow1$
such that $\tau -1=O\left(n^{-2/3}\right)$. We show that the asymptotics of the logarithmic derivative
of the Toeplitz determinant $D_{n,\nu}(t)$ generated by \eqref{eq: weight}, and equivalently the Toeplitz determinant  \eqref{TopBes}  with  entries given in terms of $I$ Bessel functions,  can be expressed in terms of the Hamiltonian of the Painlev\'{e} II
equation \eqref{Painleve2}. The asymptotics of  the recurrence coefficient in \eqref{three-term1} and  the leading
coefficient of the orthonormal polynomials with respect to the weight function  \eqref{eq: weight} are also derived.
\begin{thm}\label{thm2}
Let $n\rightarrow\infty$ and $\tau=\frac{t}{n}\rightarrow1$, such that   $\tau -1=O(n^{-2/3})$.
We have the asymptotic approximation of the logarithmic derivative of the Toeplitz determinant
associated with \eqref{eq: weight} \begin{equation}\label{eq:Dasy}
\frac{d}{dt}\log D_{n,\nu}(t)=\frac{t}{2}+ {2^{\frac {2} {3}} }{n^{-\frac 1 3}}
H\left(2^{\frac {2} {3}}n^{\frac 2 3}(\tau-1);\nu\right)+O\left(n^{-\frac 2 3}\right).
\end{equation}
Moreover, we have  the asymptotics of the leading coefficient $\gamma_n=h_n^{-1/2}$ of the $n$-th orthonormal polynomials associated with
\eqref{eq: weight}
\begin{equation}\label{eq:hasy}
\gamma_{n}=1+{2^{-\frac {1} {3}} }{n^{-\frac 1 3}}H\left(2^{\frac{2}{3}}n^{\frac 2 3}(\tau-1);\nu\right )+O\left(n^{-\frac 2 3}\right),
\end{equation}
and the logarithmic derivative of  the recurrence coefficient  in \eqref{three-term1}
\begin{equation}\label{eq:Rasy}
\frac{d}{dt}\log \pi_{n}(0)=-{2^{\frac {2} {3}} }{n^{-\frac 1 3}}u\left(2^{\frac{2}{3}}n^{\frac{2}{3}}(\tau-1);\nu\right)+O\left(n^{-\frac 2 3}\right).
\end{equation}
Here,  $u(s;\nu)$ and $H(s;\nu)$ are respectively the solution of the Painlev\'e II equation and the associated Hamiltonian with the asymptotic behaviors given in Theorem \ref{thm: PIIasy}.
For $\nu<0$ or $\nu\in\mathbb{N}$,  $u(s;\nu)$ and $H(s;\nu)$ are  pole free on the real axis and the error terms are uniform for $n^{ {2}/{3}}(\tau-1)$ in any compact subsets of the real axis. Otherwise, the error terms are uniform for $n^{3/2}(\tau-1)$ bounded and bounded away from the poles of $u(s;\nu)$ on the real axis.
\end{thm}

\begin{rem}\label{rem:PhaseT}
From  \eqref{eq:Dasy}, we see that there is a sign difference in the asymptotics of the Hamiltonian \eqref{eq:Hasyneginfty}. It is  consistent with the fact that $H(s;\nu)=H(s;-\nu)$ for $\nu\in\mathbb{N}$ as shown later in \eqref{eq: SymmH}, and the symmetry relation
\begin{equation}\label{eq:symmetry1}
D_{n,\nu}(t)=D_{n,-\nu}(t)
\end{equation}
for $\nu\in\mathbb{N}$.  The sign difference in \eqref{eq:uasyneginfty} and \eqref{eq:Hasyneginfty} may indicate that the asymptotic behaviors of  $u(s;\nu)$ and $H(s;\nu)$ are sensitive to
the parameter $\nu$ near $\mathbb{N}$. Similar phenomenon has also been observed  in the study of the asymptotics of the increasing tritronqu\'ee solutions of the second Painlev\'e equation \eqref{Painleve2} where the asymptotic behaviors  change dramatically for  the parameter $\nu$ near $\mathbb{Z}$; see \cite[Theorem 1.2]{M} and the discussion after Theorem 1.4 in \cite{M}. It would be interesting to consider the transition asymptotics of $u(s;\nu)$ and $H(s;\nu)$ as $x\to-\infty$ and $\nu\to k$ for given $k\in\mathbb{N}$ in an interrelated manner.
\end{rem}


%
The rest of the paper is arranged as follows. In Section \ref{sec2}, we formulate a Riemann-Hilbert (RH, for short) problem $Y$ for the orthogonal polynomials with respect to the weight function \eqref{eq: weight}, and derive  the differential identity for the Toeplitz determinant.  We then prove Theorem \ref{thm1} at the end of  this Section by relating the RH problem to that of the Painlev\'e III equation. In Section \ref{RHanalysis}, we perform the Deift-Zhou nonlinear steepest descent analysis of the RH problem for $Y$. The main concern is the construction of a local parametrix at the critical point $z=-1$, where the $\psi$-function of the Jimbo-Miwa Lax pair for the Painlev\'{e} II equation is involved.  The construction is  different from \cite{CK,CKV} where the  Flaschka-Newell Lax pair for  the Painlev\'{e} II equation was used.
In Sections \ref{section:model} and \ref{sec:Psiinfty}, we derive the asymptotics  of the one parameter family of  solutions of the  Painlev\'{e} II equation and the associated Hamiltonians as stated in Theorem \ref{thm: PIIasy} by carrying out  a nonlinear steepest descent analysis of the RH problem for the  Jimbo-Miwa Lax pair of the Painlev\'{e} II equation.
Then, the proof  of Theorem  \ref{thm2} is given in Section \ref{sectionDh}. For the convenience of
the reader, we collect the Airy and parabolic cylinder parametrix models    in the Appendix.

\section{Riemann-Hilbert problem for the orthogonal polynomials }\label{sec2}

\subsection{Riemann-Hilbert problem for orthogonal polynomials and differential identity }

\begin{rhp} \label{RHP: Y}
We look for a $2 \times 2$ matrix-valued function $Y(z;n)$ ($Y(z)$, for short) satisfying the properties.
\begin{itemize}
\item[\rm (1)] $Y(z)$ is analytic in $\mathbb{C} \setminus \Gamma$; where the contour $ \Gamma$ is shown in Figure \ref{hankelloop}.
\item[\rm (2)] $Y(z)$  satisfies the jump condition
  \begin{equation}\label{eq:Yjump}
  Y_{+}(z)=Y_{-}(z)
  \begin{pmatrix}
  1 & \frac{w(z )}{z^n} \\
  0 & 1
  \end{pmatrix},\qquad z\in \Gamma,
  \end{equation}
where the weight  function is defined in \eqref{eq: weight}.
\item[\rm (3)] As $z\to \infty$, we have
  \begin{equation}\label{eq:YInfinity}
  Y(z)=\left (I+\frac{Y_{-1}}{z}+O\left (\frac 1 {z^2}\right )\right)
 \begin{pmatrix}
 z^n & 0 \\
 0 & z^{-n}
 \end{pmatrix}.
 \end{equation}
\end{itemize}
\end{rhp}
According to \cite{FIK}, if there exists a solution to the RH problem for $Y$,   the solution is unique and given by
\begin{equation}\label{eq:YSolution}
Y(z)=Y(z;n)=
\begin{pmatrix}
\pi_n(z) &  \frac{1}{2\pi i } \int_{\Gamma}\frac{\pi_n(x) w(x)dx}{x^n(x-z)} \\[.3cm]
-h_{n-1}^{-1} \tilde{\pi}^*_{n-1}(z)&   -\frac{h_{n-1}^{-1}}{2\pi i}  \int_{\Gamma}\frac{\tilde{\pi}^*_{n-1}(x) w(x)dx}{x^n(x-z)}
\end{pmatrix},
\end{equation}
 where $\widetilde{\pi}^*_{n-1}(z)=z^{n-1}\widetilde{\pi}_{n-1}(z^{-1})$ and $\pi_n$, $\tilde{\pi}_{n-1}$  and $h_{n-1}$  are defined by \eqref{eq:Ortho}.
 \begin{rem}
 In the asymptotic analysis of the RH problem for $Y$ given in the next section, we will  transform the RH problem for $Y$ to a small-norm RH problem  for $R$ by a series of invertible transformations. The RH problem for $R$ is  solvable for sufficiently large $n$ and $\tau=\frac t n\to1$ such that $\xi=n^{2/3}(\tau-1)=O(1)$ and $x=\xi$ is not a pole of the Hasting-McLeod solution $u(x;\nu)$ of Painlev\'e II equation \eqref{Painleve2}.
 As shown in Remark \ref{rem:poles}, the Hastings-McLeod solution  $u(x;\nu)$ is pole free on the real line for $\nu<0$ or $\nu\in\mathbb{N}$, and may have poles on the real line for other parameter $\nu\in\mathbb{C}$.
Therefore, the RH problem for $Y$ is solvable under the same conditions as that for $R$ by tracing back the series of invertible transformations.
\end{rem}

For later use, we derive a differential identity for the
logarithmic derivative of the Toeplitz determinant associated with  \eqref{eq: weight}.
\begin{pro} \label{Pro: DifferentialIdentity} We have the following  differential identity
\begin{equation}\label{eq:dD}
\frac{d}{dt}\log D_{n,\nu}=-\frac{1}{2}\left((Y_{-1})_{11}+\frac{Y_{21}'(0;n+1)}{Y_{21}(0;n+1)}\right),
\end{equation}
\begin{equation}\label{eq:hY}
h_{n}=Y_{12}(0;n),
\end{equation}
where $Y_{-1}$ is given in \eqref{eq:YInfinity} and  $Y_{21}'$ denotes the derivative of $Y_{21}$ with respect to $z$.
\end{pro}
\begin{proof} The relation \eqref{eq:hY} follows directly from \eqref{eq:YSolution}. To prove \eqref{eq:dD}, we first derive from \eqref{eq: HTop} that
\begin{equation}\label{eq:DDn}
\frac{d}{dt}\log D_{n,\nu}=\sum_{k=0}^{n-1} \frac{1}{h_k}\frac{d}{dt} h_k.
\end{equation}
  Next, we calculate the derivative of $h_k$ by
using the integral representation \eqref{eq:Ortho}
\begin{align}\nonumber
\frac{d}{dt}h_k=&\frac{d}{dt} \int_{\Gamma} \pi_k(s) \tilde{\pi}_k(s^{-1}) w(s) \frac{ds}{2\pi is}\nonumber \\
=& \int_{\Gamma}\frac{d}{dt}( \pi_k(s)) \tilde{\pi}_k(s^{-1}) w(s) \frac{ds}{2\pi is}+ \int_{\Gamma}\pi_k(s)\frac{d}{dt}(  \tilde{\pi}_k(s^{-1})) w(s) \frac{ds}{2\pi is}\\
&+\int_{\Gamma}\pi_k(s) \tilde{\pi}_k(s^{-1})\frac{d}{dt} w(s) \frac{ds}{2\pi is}\nonumber \\
=&\frac{1}{2}\int_{\Gamma}\pi_k(s)s \tilde{\pi}_k(s^{-1})w(s) \frac{ds}{2\pi is}+\frac{1}{2}\int_{\Gamma}\pi_k(s) \tilde{\pi}_k(s^{-1})s^{-1}w(s) \frac{ds}{2\pi is}.\label{eq:Dh}
\end{align}
Denote
\begin{equation}\label{eq:an}
\pi_n(z)=z^n+a_{n,n-1}z^{n-1}+\cdots, \quad \tilde{\pi}_n(z)=z^n+\tilde{a}_{n,n-1}z^{n-1}+\cdots.
\end{equation}
We have
\begin{equation}\label{eq:rec}
z\pi_n(z)=\pi_{n+1}(z)+(a_{n,n-1}-a_{n+1,n})\pi_n(z)+\cdots,
\end{equation}
\begin{equation}\label{eq:rectilde}
 z\tilde{\pi}_n(z)=\tilde{\pi}_{n+1}(z)+(\tilde{a}_{n,n-1}-\tilde{a}_{n+1,n})\tilde{\pi}_n(z)+\cdots.
\end{equation}
Substituting \eqref{eq:rec} and \eqref{eq:rectilde} into \eqref{eq:Dh} yields
\begin{equation}\label{eq:Integral1}
\int_{\Gamma}\pi_k(s)s \tilde{\pi}_k(s^{-1})w(s) \frac{ds}{2\pi is}= (a_{k,k-1}-a_{k+1,k})h_{k}, \end{equation}
and
\begin{equation}\label{eq:Integral2}
 \int_{\Gamma}\pi_k(s) \tilde{\pi}_k(s^{-1})s^{-1}w(s) \frac{ds}{2\pi is}=(\tilde{a}_{k,k-1}-\tilde{a}_{k+1,k})h_k.
\end{equation}
Thus, we have
\begin{equation}\label{eq:DDnan}
\frac{d}{dt}\log D_{n,\nu}=-\frac{1}{2}(a_{n,n-1}+\tilde{a}_{n,n-1}).
\end{equation}
Comparing  \eqref{eq:YSolution} and \eqref{eq:an}, we have
\begin{equation}\label{eq:aY}
a_{n,n-1}=(Y_{-1})_{11}, \quad  \tilde{a}_{n,n-1}=(Y)_{21}'(0)/(Y)_{21}(0),
\end{equation}
where $Y_{21}'$ denotes the derivative of $Y_{21}$ with respect to $z$.
Therefore, we obtain \eqref{eq:dD}  by substituting \eqref{eq:aY} in \eqref{eq:DDnan}.
This completes the proof of the proposition.

\end{proof}

Next, we relate the RH problem for $Y$ to the known one in the literature for  Painlev\'{e} III equation; see \cite[Chapter 5.3]{FIKN}.
For this purpose, we introduce the new  independent variables
\begin{equation*}
x=\frac{t}{2}\quad \text{and}\quad \lambda=-iz,
\end{equation*}
and define
\begin{equation}\label{YtoPhi}
\Phi(\lambda,x)=e^{-\frac{(n+\nu)\pi i}{4}{\sigma_3}}Y(e^{\frac{1}{2}\pi i}\lambda)e^{\frac{ix}{2}( \lambda-\frac{1}{\lambda})\sigma_{3}}(e^{\frac{1}{2}\pi i}\lambda)^{-(\frac{n-\nu}{2})\sigma_{3}},
\end{equation}
where the branch of the function $\lambda^{-(\frac{n-\nu}{2})}$ is chosen such that $\arg \lambda\in(-\frac{3}{2}\pi,\frac{1}{2}\pi)$. 
Then, $\Phi(\lambda)=\Phi(\lambda,x)$   solves the following RH problem.
\begin{rhp}\label{RHPIII}
\item{(1)} $\Phi(\lambda)$ is analytic for $\lambda\in\mathbb{C}\setminus \{ -i\Gamma\cup (0, i\infty)\}$, where $\Gamma$ is shown in Figure \ref{hankelloop}.

\item{(2)} $\Phi(\lambda)$ satisfies the following jump condition

\begin{equation}\label{eq:Phijumps}
\Phi_{+}(\lambda)=\Phi_{-}(\lambda)\begin{pmatrix}
1 & 1 \\ 0 & 1
\end{pmatrix},\quad \lambda\in -i\Gamma~ ~\mbox{and} ~ ~\Phi_{+}(\lambda)=\Phi_{-}(\lambda)e^{-\pi i(n+\nu)\sigma_3},  \quad \lambda\in (0, i\infty).
\end{equation}

\item{(3)} As $\lambda\to \infty$, we have
\begin{equation}\label{phi-infty}
\Phi(\lambda)=\left(I+\sum_{k=1}^{\infty}\frac{\Phi_{-k}}{\lambda^{k}}\right)
e^{\frac{ix\lambda}{2}\sigma_{3}}\lambda^{\frac{n+\nu}{2}\sigma_{3}},
\end{equation}
where the branch is chosen such that $\arg \lambda\in(-\frac{3}{2}\pi,\frac{1}{2}\pi)$ and the coefficient
\begin{equation}\label{Phi-1}
\Phi_{-1}(x)=-i e^{-\frac{(n+\nu)\pi i}{4}{\sigma_3}}\begin{pmatrix}
a_{n,n-1}+\frac{x}{2} &\pi_{n+1}(0)h_{n} \\
 -\widetilde{\pi}_{n-1}(0)h^{-1}_{n-1} & -a_{n,n-1}-\frac{x}{2}
\end{pmatrix}  e^{\frac{(n+\nu)\pi i}{4}{\sigma_3}}.
\end{equation} 

\item{(4)} As $\lambda\to 0$, we have
\begin{equation}\label{phi-zero}
\Phi(\lambda)=\Phi(0)\left(I+\sum_{k=1}^{\infty}{\Phi_{k}}{\lambda^{k}}\right)
e^{-\frac{ix}{2\lambda}\sigma_{3}}\lambda^{\frac{\nu-n}{2}\sigma_{3}},
\end{equation}
where the branch is chosen such that $\arg \lambda\in(-\frac{3}{2}\pi,\frac{1}{2}\pi)$ and the coefficient
\begin{equation*}
\Phi(0)=
e^{-\frac{(n+\nu)\pi i}{4}{\sigma_3}}\begin{pmatrix}
1 & p_{n} \\ -q_{n} & 1-p_{n}q_{n}
\end{pmatrix}\pi_{n}(0)^{\sigma_{3}}e^{-\frac{(n-\nu)\pi i}{4}{\sigma_3}},
\end{equation*}
with
\begin{equation*}
p_{n}=\pi_{n}(0)h_{n},\quad q_{n}=\frac{1}{h_{n-1}\pi_{n}(0)}.
\end{equation*}
\end{rhp}

It is seen that $\Phi$ satisfies the same RH problem  for  Painlev\'{e} III equation as given in \cite[Chapter 5.3]{FIKN}.
From the RH problem \ref{RHPIII}, we have the following Lax pair
\begin{equation}\label{Phizx}
\Phi_{\lambda}=A(\lambda,x)\Phi\quad\text{and}\quad
\Phi_{x}=B(\lambda,x)\Phi,
\end{equation}
where
\begin{equation}\label{AB}
A(\lambda,x)=\frac{ix}{2}\sigma_{3}+\frac{A_{-1}}{\lambda}+\frac{A_{-2}}{\lambda^2},\quad
B(\lambda,x)=\frac{i}{2}\lambda \sigma_{3}+B_{0}+\frac{B_{-1}}{\lambda}.
\end{equation}
The coefficients in the above formula are given by
\begin{equation}\label{A-1}
A_{-1}=\begin{pmatrix}
\frac{n+\nu}{2} & -xe^{-\frac{(n+\nu)\pi i}{2}}\pi_{n+1}(0)h_{n} \\ -xe^{\frac{(n+\nu)\pi i}{2}} \widetilde{\pi}_{n-1}(0)h^{-1}_{n-1} &-\frac{(n+\nu)}{2}\end{pmatrix},
\end{equation}

\begin{equation}\label{A-2}
A_{-2}=\frac{ix}{2}\begin{pmatrix}
1-2p_{n}q_{n} & -2e^{-\frac{(n+\nu)\pi i}{2}}p_{n} \\ 2e^{\frac{(n+\nu)\pi i}{2}}q_{n}(p_{n}q_{n}-1) &2p_{n}q_{n}-1\end{pmatrix},
\end{equation}
and
\begin{equation}\label{B0}
B_{0}=\frac{1}{x}\left(A_{-1}-\frac{n+\nu}{2}\sigma_{3}\right),\quad
B_{-1}=-\frac{1}{x}A_{-2}.
\end{equation}

To complete  this section, we shall identify the Toeplitz determinant with the $\tau$-function of the Painlev\'{e} III equation, as stated in Theorem \ref{thm1}.
\subsection{Proof of Theorem \ref{thm1}}
According to \eqref{A-1},  \eqref{A-2} and  \cite[(5.3.4), (5.3.7)]{FIKN},  we find that
\begin{equation}\label{PIIIsolu}
a_n(x)=-i(A_{-1})_{12}/(A_{-2})_{12}=-\frac{\pi_{n+1}(0)}{\pi_{n}(0)}
\end{equation}
with $t=2x$, solves the Painlev\'e III equation \eqref{PIIIequ}.

Next, we consider  the Toeplitz determinant  $D_{n,\nu}$.
To this end, we derive from \eqref{three-term1}, \eqref{eq:YSolution} and \eqref{YtoPhi} the coefficient $\Phi_{1}$ in the expansion \eqref{phi-zero}
\begin{equation}\label{phi1}
\Phi_{1}=i\left(\widetilde{a}_{n,n-1}+\frac{x}{2}\right )\sigma_{3}+\begin{pmatrix} 0 & * \\ * & 0\end{pmatrix}.
\end{equation}
According to the general theory of Jimbo-Miwa-Ueno\cite[(1.11)]{JMU}, the isomonodromy $\tau$-function for the Lax pair in
\eqref{Phizx}-\eqref{B0} is defined by
\begin{eqnarray}\label{dlog}
d\log \tau(x)&=&- \Res\limits_{\lambda=0}\; {\rm{tr}} 
\left(\Phi^{-1}_{0}(\lambda)\frac{\partial\,\Phi_{0}(\lambda)}{\partial\lambda}
d\,T_{0}(\lambda)\right)\nonumber\\
&&- \Res\limits_{\lambda=\infty}\; {\rm{tr}}\left(\Phi^{-1}_{\infty}(\lambda)
\frac{\partial\,\Phi_{\infty}(\lambda)}{\partial\lambda}d\,T_{\infty}(\lambda)\right),
\end{eqnarray}
where
\begin{equation*}
d\,T_{0}(\lambda)=-\frac{i}{2\lambda}\sigma_{3}dx,\quad d\,T_{\infty}(\lambda)=\frac{i\lambda}{2}\sigma_{3}dx.
\end{equation*}
Substituting \eqref{Phi-1} and \eqref{phi1} into \eqref{dlog}, we arrive at the equation
\begin{eqnarray}\label{dlogtau}
\frac{d}{dx}\log \tau(x)&=&\frac{i}{2}\; {\rm{tr}}(\Phi_{1}\sigma_{3}-\Phi_{-1}\sigma_{3})\nonumber\\
&=&\frac{i}{2}(2i\widetilde{a}_{n,n-1}+2ia_{n,n-1}+2ix)\nonumber\\
&=&-\widetilde{a}_{n,n-1}-a_{n,n-1}-x.
\end{eqnarray}
From \eqref{eq:DDnan} and \eqref{dlogtau}, we obtain \eqref{fixed-dn}.
This completes the proof of Theorem  \ref{thm1}.



 \section{Asymptotics of  the Riemann-Hilbert problem for orthogonal polynomials}\label{RHanalysis}
 In this section, we perform the Deift-Zhou nonlinear steepest descent analysis of the RH problem for
$Y$.
The analysis includes a series of invertible transformations $Y\rightarrow \widehat{Y}\rightarrow T\rightarrow S\rightarrow R$ such that   the jump matrices in the final RH problem for $R$ are uniformly close to the identity matrix; \cite{D1,DKMVZ1,DKMVZ2,DZ1}.

We first modify the contour in RH problem  for $Y$ and define 
\begin{equation}\label{eq:Yhat}
  \widehat{Y}(z)=\left\{ \begin{array}{ll}
                      Y(z), & \quad z\in \Omega_{-}\cup \Omega_{0} \\
                     Y(z) \begin{pmatrix}1&-\frac{w(z)}{z^{n}}\\0 &1\end{pmatrix}, & \quad z\in \Omega_{+}\backslash\Omega_{0},
                    \end{array}
\right.
 \end{equation}
where the regions and contours are shown in Figure \ref{hankelcontour}. Then we obtain the following RH problem for $\widehat{Y}$.

\begin{figure}[h]
  \centering
  \includegraphics[width=8cm,height=4cm]{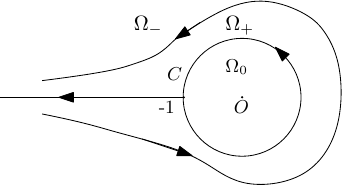}\\
  \caption{Contours for the transformation  $Y\to \widehat{Y}$}\label{hankelcontour}
\end{figure}

 \begin{rhp}\label{RHP: Yhat}

\item[(1)] $\widehat{Y}(z)$ is analytic in $\mathbb{C}\setminus \Sigma_{\widehat{Y}} $, where the contour  $\Sigma_{\widehat{Y}}= (-1,-\infty)\cup C$ are depicted in Figure \ref{Yhatcontour}, of which $C$ is the unit circle centered at the origin.
\item[(2)] $\widehat{Y}(z)$  satisfies $\widehat{Y}_{+}(z)=\widehat{Y}_{-}(z)J_{\widehat{Y}}$, where
  \begin{equation}\label{eq:Yhatjump}
 J_{\widehat{Y}}(z)=\left\{
\begin{aligned}
&\begin{pmatrix}1 & \frac{w(z)}{z^{n}}\\ 0 & 1 \end{pmatrix},\quad &z&\in C,\\
&\begin{pmatrix}1 & \frac{w_{+}(z)-w_{-}(z)}{z^{n}}\\ 0 & 1 \end{pmatrix},\quad &z&\in (-1,-\infty),
\end{aligned}
\right.
  \end{equation}
  with $w_{+}(z)-w_{-}(z)=|z|^{\nu}e^{\frac{t}{2}(z+\frac{1}{z})}(e^{\pi i \nu}-e^{-\pi i \nu})$.

  \item[(3)] As $z\to \infty$, we have
  \begin{equation}\label{eq:UInfinity}
  \widehat{Y}(z)=\left (I+\frac{\widehat{Y}_{-1}}{z}+O\left (\frac 1 {z^2}\right )\right )z^{n\sigma_{3}}.
 \end{equation}
\end{rhp}
\begin{figure}[h]
  \centering
  \includegraphics[width=8cm,height=2.9cm]{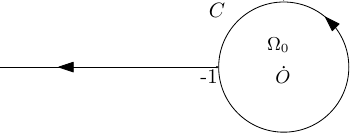}\\
  \caption{The contours $\Sigma_{\widehat{Y}}$ of RH problem for $\widehat{Y}$}\label{Yhatcontour}
\end{figure}

\subsection{Normalization: $\widehat{Y}\to T$}
To normalize the large-$z$ behavior of $\widehat{Y}(z)$, we introduce the transformation
\begin{equation}\label{eq:YhatT}
  T(z)=\left\{ \begin{array}{ll}
                      \widehat{Y}(z)e^{\frac{t}{2z}\sigma_3}z^{-n\sigma_3}, & \hbox{$|z|>1$,}\\
                     \widehat{Y}(z)e^{\frac{t}{2}z\sigma_3}, & \hbox{$|z|<1$.}
                    \end{array}
\right.
 \end{equation}

\begin{rhp} \label{RHP: T}
The function $T(z)$ defined in \eqref{eq:YhatT} satisfies the following RH problem.
\begin{itemize}
\item[\rm (1)] $T(z)$ is analytic in $\mathbb{C} \setminus \Sigma_{\widehat{Y}}$.
\item[\rm (2)] $T(z)$  satisfies the jump condition
  \begin{equation}\label{eq:TJump}
  T_{+}(z)=T_{-}(z)\left\{
  \begin{aligned}
  &\begin{pmatrix}e^{n\phi(z)} & z^{\nu} \\0 & e^{-n\phi(z)}\end{pmatrix},\quad &z&\in C,\\
  &\begin{pmatrix}1 & |z|^{\nu}e^{n\phi(z)}(e^{\pi i \nu}-e^{-\pi i \nu})\\ 0 & 1 \end{pmatrix},\quad &z&\in (-1,-\infty),
  \end{aligned}
\right.
  \end{equation}
  where
   \begin{equation}\label{eq:phi}
  \phi(z)=\frac{t}{2n}(z-z^{-1})+\log z.
 \end{equation}
\item[\rm (3)] As $z\to \infty$, we have
  \begin{equation}\label{eq:TInfinity}
  T(z)=I+O\left (\frac 1 {z}\right ).
 \end{equation}
\end{itemize}
\end{rhp}

\subsection{Deformation: $T\to S$}
It is seen from  \eqref{eq:TJump}  that the diagonal entries of the jump matrix for $T$  are highly oscillating for $n$ large.
 To transform the oscillating entries  to exponential decay ones on certain contours, we introduce
 the second transformation $T\to S$.
The transformation is based on the following factorization of jump matrix
\begin{equation}\label{eq:Factorization}
  \begin{pmatrix}
  e^{n\phi(z)} & z^{\nu} \\
  0 & e^{-n\phi(z)}
  \end{pmatrix}=   \begin{pmatrix}
  1 &0 \\
  z^{-\nu} e^{-n\phi(z)} & 1
  \end{pmatrix}  \begin{pmatrix}
  0& z^{\nu} \\
  -z^{-\nu} &0
  \end{pmatrix} \begin{pmatrix}
  1 &0 \\
  z^{-\nu} e^{n\phi(z)} & 1
  \end{pmatrix}.
  \end{equation}
  We introduce the transformation
  \begin{equation}\label{eq:TS}
  S(z)=\left\{ \begin{array}{lll}
                   T(z) \begin{pmatrix}
  1 & 0 \\
 z^{-\nu} e^{-n\phi(z)} & 1
  \end{pmatrix} , & \hbox{for $z\in  \Omega_E$,}\\[.5cm]
  T(z)\begin{pmatrix}
  1 & 0 \\
  - z^{-\nu} e^{n\phi(z)} &1
  \end{pmatrix}  ,   &   \hbox{for $z\in  \Omega_I$,}\\
     T(z), & \hbox{otherwise.} \\
  \end{array}
\right.
 \end{equation}
Here $\Omega_E$ and $\Omega_I$ denote some lens-shaped regions  outside and inside the unit circle,
 which are shown in Figure \ref{hankelcontourS}.

\begin{rhp} \label{RHP: S}
The function $S(z)$ defined in \eqref{eq:TS} satisfies the following RH problem.
\begin{itemize}
\item[\rm (1)] $S(z)$ is analytic in $\mathbb{C} \setminus \Sigma$, where $\Sigma=\Sigma_E\cup C \cup \Sigma_I\cup(-1,-\infty)$, with $\Sigma_E$ and $\Sigma_I$ the boundary of the lens shaped regions  $\Omega_E$ and $\Omega_I$ as indicated in Figure \ref{hankelcontourS}.
\item[\rm (2)] $S(z)$  satisfies the jump condition
  \begin{equation}\label{eq:SJump}
  S_{+}(z)=S_{-}(z)
  \left\{ \begin{array}{lll}
                    \begin{pmatrix}
  1 & 0 \\
 z^{-\nu} e^{-n\phi(z)} & 1
  \end{pmatrix} , & \hbox{for $z\in  \Sigma_E$,}\\[.5cm]
   \begin{pmatrix}
  0& z^{\nu} \\
  -z^{-\nu} &0
  \end{pmatrix} , & \hbox{for $z\in C$,} \\[.5cm]
 \begin{pmatrix}
  1 & 0 \\
   z^{-\nu} e^{n\phi(z)} &1
  \end{pmatrix}  ,   &   \hbox{for $z\in  \Sigma_I$,}\\[.5cm]
  \begin{pmatrix}1 & |z|^{\nu}e^{n\phi(z)}(e^{\pi i \nu}-e^{-\pi i \nu})\\ 0 & 1 \end{pmatrix},& \hbox{for $z\in  (-1,-\infty)$},\\
  \end{array}
\right.
  \end{equation}
  with $\phi(z)$ defined in \eqref{eq:phi}.
\item[\rm (3)] As $z\to \infty$, we have
  \begin{equation}\label{eq:SInfinity}
  S(z)=I+O\left (\frac 1 {z}\right ).
 \end{equation}
\end{itemize}
\end{rhp}
\begin{figure}[ht]
  \centering
  \includegraphics[width=10.2cm,height=5.5cm]{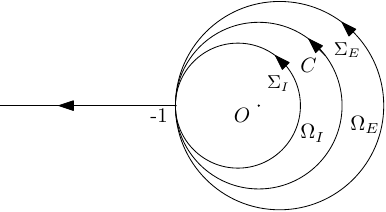}\\
  \caption{Contours and regions for the RH problem for $S(z)$}\label{hankelcontourS}
\end{figure}

For $t= n\tau$, $\tau\to 1$,   it follows  from \eqref{eq:phi} that the $\phi$-function possesses the following properties
\begin{equation}\label{eq:phiSgn}
\left\{ \begin{array}{lll}
\Re  \phi(z)=0,  &z\in C,& \\
  \Re \phi(z)>0, &z\in \Sigma_E, & |z+1|>\delta,\\
     \Re  \phi(z)<0, &z\in \Sigma_I,&  |z+1|>\delta;
  \end{array}
\right.
\end{equation}
see  \cite{BDJ}.

\subsection{Global parametrix: $N$}
It is readily  seen from \eqref{eq:phiSgn}
  that on $\Sigma_I\cup\Sigma_E$ and bounded away from $z=-1$,
 the jump matrices for $S$ tend to the identity matrix exponentially fast as $n\to\infty$.
 Then, we arrive at the following approximate RH problem for $n$ large.

 \begin{rhp} \label{RHP: N}
We look for a $2 \times 2$ matrix-valued function $N(z)$ satisfying the following properties.
\begin{itemize}
\item[\rm (1)] $N(z)$ is analytic in $\mathbb{C} \setminus C$, where $C$ is the unit circle oriented counterclockwise.
\item[\rm (2)] $N(z)$  satisfies the jump condition
  \begin{equation}\label{eq:Njump}
  N_{+}(z)=N_{-}(z)
 \begin{pmatrix}
  0& z^{\nu} \\
  -z^{-\nu} &0
  \end{pmatrix} ,\qquad z\in C.
  \end{equation}
  \item[\rm (3)] As $z\to \infty$, we have
  \begin{equation}\label{eq:NInfinity}
  N(z)=I+O\left (\frac 1 {z}\right ).
 \end{equation}
\end{itemize}
\end{rhp}

The solution to the  RH problem for $N(z)$ can be constructed by elementary functions as follows:
\begin{equation}\label{eq:NSolution}
  N(z)=\left\{ \begin{array}{ll}
                      \left(\frac{z+1}{z}\right)^{\nu\sigma_3}, & \hbox{$|z|>1$,} \\
                     (z+1)^{\nu \sigma_3}\begin{pmatrix}
 0 & 1 \\
  -1 & 0
  \end{pmatrix}, & \hbox{$|z|<1$.}
                    \end{array}
\right.
 \end{equation}
 Here, the branches for the power functions are chosen such that $\arg(z)\in(-\pi, \pi)$ and  $\arg(z+1)\in(-\pi, \pi)$.
\subsection{Local parametrix and the model Riemann-Hilbert problem for $\Psi$}
In this subsection, we intend to construct a local parametrix in a neighborhood of  $z=-1$.
\begin{rhp} \label{RHP: P}
We look for a $2 \times 2$ matrix-valued function $P(z)$ satisfying the following properties.
\begin{itemize}
\item[\rm (1)] $P(z)$ is analytic in $U(-1,\delta)\setminus \Sigma$, where $U(-1,\delta)$ is a neighborhood of  $z=-1$.
\item[\rm (2)] $P(z)$  satisfies the same  jump condition as $S(z)$ on $U(-1,\delta)\cap \Sigma$.
   \item[\rm (3)] On the boundary of $U(-1,\delta)$, $P(z)$ satisfies the matching condition
  \begin{equation}\label{eq:Matching}
P(z)N^{-1}= n^{-\frac{ \nu}{3}\sigma_3}\left(I+O(1/n^{1/3})\right)n^{\frac{ \nu}{3}\sigma_3}. 
 \end{equation}
\end{itemize}
\end{rhp}

Let
\begin{equation}\label{eq:tildeP}
  \tilde {P}(z)=\left\{ \begin{array}{lll}
                     P(z), &  z\in U(-1,\delta), &|z|>1 ,  \\
                    P(z) \begin{pmatrix}
 0 & -z^{\nu} \\
  z^{-\nu} & 0
  \end{pmatrix}, &  z\in U(-1,\delta) , &|z|<1.
                    \end{array}
\right.
 \end{equation}
It is seen from \eqref{eq:SJump} that $\tilde{P}(z)$ satisfies the jump relations
\begin{equation}\label{eq:tildePJump}
   \tilde {P}_{+}(z)= \tilde {P}_{-}(z)
  \left\{ \begin{array}{lll}
                    \begin{pmatrix}
  1 & 0 \\
 z^{-\nu} e^{-n\phi(z)} & 1
  \end{pmatrix} , & \hbox{for $z\in U(-1,\delta)\cap \Sigma_E$,}\\[.5cm]
 \begin{pmatrix}
  1 &  - z^{\nu} e^{n\phi(z)} \\0 &1\end{pmatrix}  ,   &   \hbox{for $z\in U(-1,\delta)\cap \Sigma_I$,}\\[.5cm]
  \begin{pmatrix}1 & |z|^{\nu}e^{n\phi(z)}(e^{\pi i \nu}-e^{-\pi i \nu})\\ 0 & 1 \end{pmatrix},& \hbox{for $z\in  U(-1,\delta)\cap(-1,-\infty)$},\\[.5cm]
  \begin{pmatrix}e^{-2\pi i \nu} & 0\\ 0 & e^{2\pi i \nu} \end{pmatrix},& \hbox{for $z\in  U(-1,\delta)\cap(-1,0)$}.\\
  \end{array}
\right.
  \end{equation}
We seek a solution to the RH problem for $P$ of the following form for $|z+1|<\delta$
 \begin{equation}\label{solu:P}
 P(z)=E(z)\Psi(n^{1/3}f(z), n^{2/3}s(t,z))\left\{\begin{array}{ll} z^{-\frac{\nu}{2}\sigma_{3}}e^{-\frac{1}{2}\pi i \nu\sigma_{3}} e^{-\frac{n}{2}\phi(z)\sigma_3}M(z), & \arg z\in(0,\pi),\\[.2cm]
 z^{-\frac{\nu}{2}\sigma_{3}}e^{\frac{1}{2}\pi i \nu\sigma_{3}} e^{-\frac{n}{2}(\phi(z)+2\pi i)\sigma_3}M(z), & \arg z\in(-\pi,0),
  \end{array}
\right.
  \end{equation}
where
 \begin{equation}\label{eq:M}
 M(z)=\left\{\begin{array}{lll}
 I, &z\in U(-1,\delta),& |z|>1, \\
 \begin{pmatrix}0&z^{\nu}\\-z^{-\nu}&0\end{pmatrix}, &z\in U(-1,\delta),& |z|<1,
 \end{array}
\right.
 \end{equation}
and the branch cut for $z^{\nu/2}$ is chosen such that $\arg z\in(-\pi,\pi)$.
The function $f(z)$ is defined as  
 \begin{equation}\label{eq:f}
f(z)=\left(\frac{3}{2}\phi_0(z)\right)^{1/3}\sim 2^{-2/3}(z+1), \quad z\to-1,
 \end{equation} and serves as a conformal mapping at $z=-1$,
 where
 \begin{equation*}
 \phi_0(z)=\frac{1}{2}(z-z^{-1})+\log(z)\mp \pi i~~~~\mbox{for}~~ \pm \arg z\in(-\pi,\pi) .
 \end{equation*}
 \begin{rhp} \label{RHP: Psi}
The function $\Psi(z)=\Psi(z,x)$  satisfies the following model RH problem.
\begin{itemize}
\item[\rm (1)] $\Psi(z)$ is analytic in $\mathbb{C} \setminus  \cup_{k=1}^5 ~\Gamma_k$, where
$$\Gamma_1=e^{\frac{\pi}{3}i}\mathbb{R}^{+},~ \Gamma_2=e^{\frac{2\pi}{3}i}\mathbb{R}^{+},~\Gamma_3=e^{\pi i}\mathbb{R}^{+},~\Gamma_4=e^{-\frac{2\pi}{3}i}\mathbb{R}^{+}, ~\Gamma_5=e^{-\frac{\pi}{3}i}\mathbb{R}^{+},$$
as depicted in Figure \ref{ModelRHP}.
\item[\rm (2)] For $z\in  \cup_{k=1}^5 ~\Gamma_k$, $\Psi(z)$  satisfies the jump condition
  \begin{equation}\label{eq:PsiJump}
  \Psi_{+}(z)=\Psi_{-}(z)J_{\Psi}(z), 
  \end{equation}
where
 \begin{equation}\label{eq:Sk}
   J_{\Psi}(z)=
  \left\{ \begin{array}{lll}
  \begin{pmatrix}1 &  e^{-\pi i \nu} \\0 &1\end{pmatrix}  ,   &   z\in\Gamma_{1},\\[.5cm]
 \begin{pmatrix}1 & 0 \\-e^{\pi i \nu} & 1\end{pmatrix} , & z\in\Gamma_{2},\\[.5cm]
 \begin{pmatrix}e^{-2\pi i \nu} & e^{\pi i \nu}-e^{-\pi i \nu}\\ 0 & e^{2\pi i \nu} \end{pmatrix},& z\in\Gamma_{3},\\[.5cm]
     \begin{pmatrix}1 & 0 \\e^{-\pi i \nu} & 1\end{pmatrix} , & z\in\Gamma_{4},\\[.5cm]
     \begin{pmatrix}1 &  -e^{\pi i \nu} \\0 &1\end{pmatrix}  ,   &   z\in\Gamma_{5}.\\[.5cm]
  \end{array}
\right.
  \end{equation}
\item[\rm (3)] As $z\to \infty$, we have
  \begin{equation}\label{PsiInfty}
  \Psi(z)=\left(I+\frac{\Psi_{1}(x;\nu)}{z}+O\left (\frac 1 {z^2}\right )\right) \exp\left(\left(\frac{1}{3}z^3+\frac{x}{2}z\right)\sigma_3\right)z^{\nu\sigma_3},
 \end{equation}
 with $\arg z\in(-\pi,\pi)$.
\end{itemize}
\end{rhp}
\begin{rem}\label{rem}\label{rem:PsiSolution}
The solution $\Psi(z)$ of the model RH problem can be constructed by using the $\psi$-functions  of the Jimbo-Miwa-Garnier Lax pair for the general
Painlev\'e II equation \eqref{Painleve2};  see \cite[Proposition 5.3]{FIKN} and \cite{JM}.   It is worth mentioning that the $\psi$-functions are  entire functions in the variable $z$. Conversely, the solution $u(x;\nu)$ of  \eqref{Painleve2} and  the associated Hamiltonian $H(x;\nu)$
can be expressed in terms of the solution of the model RH problem
 \begin{equation}\label{eq:psi1}
H(x;\nu)=-(\Psi_{1})_{11}(x;\nu)=(\Psi_{1})_{22}(x;\nu), \quad u(x;\nu)=-\frac{d}{dx}\log(\Psi_{1})_{12}(x;\nu)
\end{equation}
with $\Psi_{1}(x;\nu)$ given in \eqref{PsiInfty}. Moreover, the  model RH problem $\Psi(z,x)$ is solvable if and only if $x$ is not a pole of $u(x;\nu)$.
\end{rem}
\begin{figure}[h]
  \centering
  \includegraphics[width=5cm,height=5.6cm]{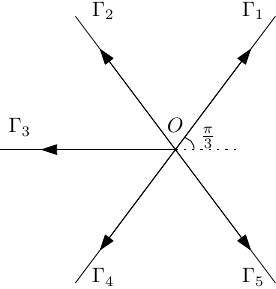}\\
  \caption{Contours and regions for the model RH problem} \label{ModelRHP}
\end{figure}
To fulfill the matching condition \eqref{eq:Matching}, we define
 \begin{equation}\label{eq:st}
s(t,z)=\frac{\phi(z)-\phi_0(z)\mp \pi i}{f(z)}=\frac{\tau-1}{2}\frac{z-z^{-1}}{f(z)}
=2^{2/3}(\tau-1) \left(1+\sum_{k=1}^{\infty}c_k(z+1)^k\right),
 \end{equation}
 for $\pm \arg z\in(0,\pi)$ with $\tau=\frac{t}{n}$.
Now we choose
 \begin{equation}\label{eq:E}
E(z)=n^{-\frac{ \nu}{3}\sigma_3}\left(\frac{z+1}{f(z)}\right)^{\nu \sigma_3}e^{\frac{n\pi i}{2}\sigma_{3}}z^{-\frac{\nu}{2}\sigma_{3}}e^{\pm\frac{1}{2}\pi i\nu\sigma_{3}}, \quad\pm \arg z\in(0,\pi).
 \end{equation}

 \begin{pro}\label{Pro:Match}Let $s(t,z)$ be given by \eqref{eq:st} and $n^{2/3}s(t,z)$ be bounded away from the poles of $u(x;\nu)$,   the matching condition \eqref{eq:Matching} is fulfilled for general parameter $\nu\in \mathbb{C}$.
 \end{pro}
 \begin{proof}
  According to Remark \ref{rem:PsiSolution},  the solution $\Psi\left(n^{1/3}f(z), n^{2/3}s(t,z)\right)$ exists for $n^{2/3}s(t,z)$  bounded away from the poles of $u(x;\nu)$.  It follows from \eqref{eq:st} that
  \begin{equation}\label{eq:theta}
\frac{1}{3} f(z)^3+\frac{1}{2}f(z)s(t,z) =\frac{1}{2}(\phi(z)\mp\pi i),~ \pm\arg z\in(0,\pi).
 \end{equation}
 This, together with \eqref{PsiInfty} and \eqref{eq:E}, implies that
   \begin{equation}\label{eq:Matching-2}
 P(z)N^{-1}=n^{-\frac{ \nu}{3}\sigma_3}\left(I+\frac{G(z)\Psi_1\left(n^{2/3}s(t,z)\right)
 G(z)^{-1}}{n^{1/3}f(z)}+O\left(1/n^{2/3}\right)\right)n^{\frac{ \nu}{3}\sigma_3},
  \end{equation}
 where
 \begin{equation}\label{eq:G}G(z)=z^{-\frac{ \nu}{2}\sigma_3}\left(\frac{z+1}{f(z)}\right)^{\nu \sigma_3}e^{\frac{n\pi i}{2}\sigma_{3}}e^{\pm\frac{1}{2}\pi i \nu\sigma_{3}}
 \end{equation}
 for $\pm\arg z\in(0,\pi)$. Therefore, we have the matching condition \eqref{eq:Matching}.
 \end{proof}
\subsection{Final transformation $S\to R$ and error estimate}

We define
\begin{equation}\label{eq:R}
  R(z)=\left\{ \begin{array}{ll}
                     n^{\frac{ \nu}{3}\sigma_3} S(z)N(z)^{-1}n^{-\frac{ \nu}{3}\sigma_3}, & \hbox{$|z+1|>\delta$,} \\
                   n^{\frac{ \nu}{3}\sigma_3} S(z)P(z)^{-1}n^{-\frac{ \nu}{3}\sigma_3}, & \hbox{$|z+1|<\delta$.}
                    \end{array}
\right.
 \end{equation}
Using \eqref{eq:Matching} and \eqref{eq:R},  we see that $R(z)$ satisfies the following RH problem.
\begin{rhp} \label{RHP: R}
The function $R(z)$ defined by \eqref{eq:R} satisfies the following properties.
\begin{itemize}
\item[\rm (1)] $R(z)$ is analytic in $\mathbb{C} \setminus\Sigma_R$, where the remaining contour $\Sigma_R$ is illustrated in Figure \ref{hankelcontourR}.
\item[\rm (2)] $R(z)$  satisfies the same  jump condition
\begin{equation}\label{eq:RJump0}
R_{+}(z)= R_{-}(z)J_R(z),
 \end{equation}
with
\begin{equation}\label{eq:RJump1}
 J_R(z)=I+O(1/n^{1/3}), \quad |z+1|=\delta,
 \end{equation}
 and
 \begin{equation}\label{eq:RJump2}
 J_R(z)=I+O(e^{-cn}),
 \end{equation}
 for z on the other jump contours.
   \item[\rm (3)] Near infinity, we have
  \begin{equation}\label{eq:RInfty}
 R(z)=I+O(1/z).
 \end{equation}
\end{itemize}
\end{rhp}
\begin{figure}[h]
  \centering
  \includegraphics[width=10cm,height=5.5cm]{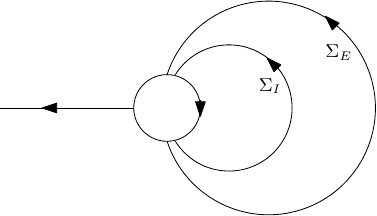}\\
  \caption{ Contours  for the RH problem for $R(z)$}\label{hankelcontourR}
\end{figure}
\begin{pro}\label{Pro:REst}
For $|\tau-1|=O(n^{-2/3})$,  we have the estimate
\begin{equation}\label{eq:REst}
 R(z)=I+\frac{R_1}{n^{1/3}}+ O\left(\frac 1{n^{2/3}}\right),
 \end{equation}
where the error term is uniform for $z$ in the complex plane.
Here
\begin{equation}\label{eq:R1O}
 R_1(z)=\frac{2^{\frac{2}{3}}}{1+z} e^{\frac{ n\pi i}{2}\sigma_3}2^{\frac{2}{3}\nu \sigma_3}\Psi_1(2^{\frac{2}{3}}n^{2/3}(\tau-1))2^{-\frac{2}{3}\nu \sigma_3}e^{-\frac{n\pi i }{2}\sigma_3}
 \end{equation}
for $|1+z|>\delta$.
 \end{pro}
 \begin{proof}
 It is seen from  \eqref{eq:Matching-2} and \eqref{eq:R} that
 \begin{equation}\label{eq:JR}
 J_{R}(z)=I+\frac{J_1(z)}{n^{1/3}}+O\left(\frac 1{n^{2/3}}\right),
 \end{equation}
for $|1+z|=\delta$.
Here
 \begin{equation}\label{eq:J1}
 J_{1}(z)=G(z)\frac{\Psi_1(n^{2/3}s(t,z))}{f(z)}G(z)^{-1}, 
 \end{equation}
 with $G(z)$ given in \eqref{eq:G}. Accordingly we have the expansion \eqref{eq:REst}.
 Substituting  \eqref{eq:REst} and \eqref{eq:JR} into \eqref{eq:RJump0}, we obtain
\begin{equation}\label{eq:R1relation}
 R_{1+}(z)= R_{1-}(z)+J_1(z) .\end{equation}
 This, together with the behavior $ R_{1}(z)=O(1/z)$, implies
 \begin{equation}\label{eq:R1}
 R_1(z)=\frac{1}{2\pi i} \oint_{|x+1|=\delta}\frac{J_1(x)}{x-z}dx,
 \end{equation}
 where the integral contour is oriented clockwise.
 Therefore, we have \eqref{eq:R1O} by \eqref{eq:Matching-2}, \eqref{eq:J1} and \eqref{eq:R1}.
 This completes the proof of the proposition.
  \end{proof}
\section{Asymptotics for $\Psi$ as $x\rightarrow+\infty$ }\label{section:model}
 We begin with the following re-scaling of variable
 \begin{equation}\label{rescaling+}
 \Phi(z)=x^{-\frac{\nu}{2}\sigma_3}\Psi(x^{\frac{1}{2}}z;x) \left\{\begin{array}{ll} I, & \arg z\in(-\frac{1}{2}\pi,\pi),\\
 e^{2\pi  i \nu\sigma_{3}} , & \arg z\in(-\pi,-\frac{1}{2}\pi).
  \end{array}
\right.
 \end{equation}
 It is seen from the RH problem \ref{RHP: Psi} that  $\Phi(z)$ satisfies the following RH problem.
\begin{rhp}The function $\Phi(z)$ defined by \eqref{rescaling+} satisfies the following properties.
\begin{itemize}
\item[\rm (1)] $ \Phi(z)$ is analytic in $\mathbb{C} \setminus \cup_{k=1}^6 \pi_{k}$, where the contours $\cup_{k=1}^6 \pi_{k}$ are shown in Figure \ref{contourPhi}.
\item[\rm (2)] $ \Phi(z)$  satisfies $ \Phi_{+}(z)= \Phi_{-}(z)J_{\Phi}(z)$, for $z\in  \cup_{k=1}^6 \pi_{k} $,  where
\begin{equation}\label{JumptildePsi}
J_{\Phi}(z)=
  \left\{ \begin{array}{lll}
  \begin{pmatrix}1 &  e^{-\pi i \nu} \\0 &1\end{pmatrix}  ,   &   z\in\pi_{1},\\[.5cm]
 \begin{pmatrix}1 & 0 \\-e^{\pi i \nu} & 1\end{pmatrix} , &  z\in\pi_{2},\\[.5cm]
 \begin{pmatrix}1 & e^{-\pi i \nu}-e^{-3\pi i \nu}\\ 0 & 1 \end{pmatrix},&  z\in\pi_{3},\\[.5cm]
     \begin{pmatrix}1 & 0 \\e^{3\pi i \nu} & 1\end{pmatrix} , &  z\in\pi_{4},\\[.5cm]
     \begin{pmatrix}1 &  -e^{\pi i \nu} \\0 &1\end{pmatrix}  ,   &    z\in\pi_{5},\\[.5cm]
     \begin{pmatrix} e^{-2\pi i \nu} & 0 \\ 0 &e^{2\pi i \nu}\end{pmatrix}  ,   &    z\in\pi_{6}.\\[.5cm]
  \end{array}
\right.
  \end{equation}
  \item[\rm (3)] As $z\to \infty$, we have
 \begin{equation}\label{rescaling+:infty}
 \Phi(z)=\left(I+\frac{\Phi_{1}}{z}+O\left(\frac{1}{z^2}\right)\right)e^{t\theta(z)\sigma_3}z^{\nu\sigma_3},
\end{equation}
where $t=x^{\frac{3}{2}}, \theta(z)=\frac{1}{3}z^3+\frac{1}{2}z$ and  $ \arg z\in(-\frac{1}{2}\pi,\frac{3}{2}\pi)$.
\end{itemize}
\end{rhp}
 \begin{figure}[h]
  \centering
  \includegraphics[width=5cm,height=5.6cm]{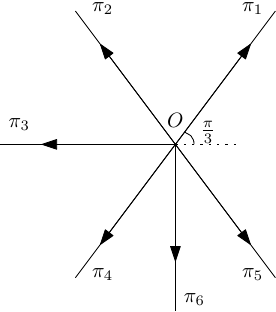}\\
  \caption{Contours and regions of the RH problem for $\Phi$ } \label{contourPhi}
\end{figure}

To normalize the asymptotic behavior of $\Phi(z)$ at infinity, we introduce the $g$-function
\begin{equation}\label{g-function}
g(z)=\frac{1}{3}(z^2+1)^{\frac{3}{2}},
 \end{equation}
where the branches are taken such that $\arg(z\pm i)\in(-\frac{\pi}{2},\frac{3\pi}{2})$. Therefore, we have
$$g_{+}(z)+g_{-}(z)=0, \quad z\in(-i,i).$$
A straightforward computation gives
\begin{equation}\label{eq:g}
g(z)=\frac{1}{3}z^3+\frac{1}{2}z+\frac{1}{8}z^{-1}+O(z^{-3}),\quad z\rightarrow\infty.
\end{equation}
\begin{figure}[h]
  \centering
  \includegraphics[width=6cm,height=7cm]{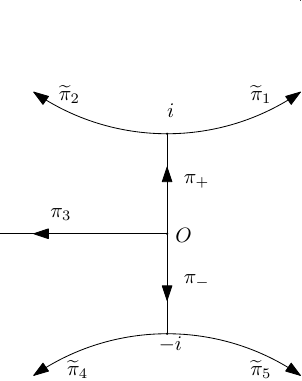}\\
  \caption{The jump curves $\Sigma_{{U}}$ of the RH problem for ${U}(z)$}\label{deform1}
\end{figure}

The second transformation is now defined as
\begin{equation}\label{eq:U}
U(z)=2^{\nu\sigma_{3}}\Phi(z)e^{-tg(z)\sigma_{3}}\varphi(z)^{-\nu\sigma_{3}},
\end{equation}
 where $\varphi(z)=z+\sqrt{z^2+1}$  is a conformal mapping from $\mathbb{C}\setminus [-i, i]$ onto  $|\varphi(z)|>1$, and the branches are taken such that $\arg(z\pm i)\in(-\frac{\pi}{2},\frac{3\pi}{2})$ and $\arg\varphi(z)\in(-\frac{\pi}{2},\frac{3\pi}{2})$. 
 This implies that
\begin{equation}\label{varphijump}
\left\{\begin{aligned}
&\left(\varphi(z)^{\nu}\right)_+\left(\varphi(z)^{\nu}\right)_-=e^{\pi i\nu}, \quad z\in(-i,i),\\
&\left(\varphi(z)^{\nu}\right)_+=\left(\varphi(z)^{\nu}\right)_-e^{2\pi i\nu}, \quad z\in(-i\infty,-i).
\end{aligned}\right.
\end{equation}

\begin{rhp}
$U(z)$ satisfies the following RH problem.
\begin{itemize}
\item[\rm (1)] $U(z)$ is analytic in $\mathbb{C} \setminus \Sigma_{U}$, where $\Sigma_{{U}}$ is shown in Figure \ref{deform1}, and where we have deformed
$\pi_{k}, k=1,2,4,5$ into the steepest descent curves $\widetilde{\pi}_{k}, k=1,2,4,5$ of the $g$-function $g(z)$ .
\item[\rm (2)] ${U}(z)$  satisfies ${U}_{+}(z)={U}_{-}(z)J_{{U}}(z)$, where
  \begin{equation}\label{eq:Ujump}
 J_{{U}}(z)=\left\{
\begin{aligned}
&\begin{pmatrix}1&\varphi(z)^{2\nu}e^{2tg(z)-\pi i\nu}\\0&1 \end{pmatrix},\quad &z&\in\widetilde{\pi}_{1},\\
&\begin{pmatrix}1&0\\-\varphi(z)^{-2\nu}e^{-2tg(z)+\pi i\nu}&1 \end{pmatrix},\quad &z&\in\widetilde{\pi}_{2},\\
&\begin{pmatrix}1& (e^{-\pi i \nu}-e^{-3\pi i \nu})\varphi(z)^{2\nu}e^{2tg(z)}\\ 0 & 1\end{pmatrix},\quad &z&\in\pi_{3}.\\
&\begin{pmatrix}1&0\\ \varphi(z)^{-2\nu}e^{-2tg(z)+3\pi i\nu}&1 \end{pmatrix},\quad &z&\in\widetilde{\pi}_{4},\\
&\begin{pmatrix}1&-\varphi(z)^{2\nu}e^{2tg(z)+\pi i\nu}\\0&1 \end{pmatrix},\quad &z&\in\widetilde{\pi}_{5},\\
&\begin{pmatrix}0 & 1\\ -1 & e^{-\pi i\nu}\varphi(z)^{2\nu}_{+}e^{2tg_{+}(z)} \end{pmatrix},\quad &z&\in\pi_{+},\\
&\begin{pmatrix}e^{-3\pi i \nu}\varphi(z)^{2\nu}_{+}e^{2tg_{+}(z)} & -1\\ 1 & 0 \end{pmatrix},\quad &z&\in\pi_{-}.
\end{aligned}
\right.
  \end{equation}
  \item[\rm (3)] As $z\to \infty$, we have
  \begin{equation}\label{eq:UInfty}
  U(z)=I+O\left(\frac 1 {z}\right).
 \end{equation}
\end{itemize}
\end{rhp}
According to the signature of $\Re g(z)$, it is readily seen that the jump matrices on $\widetilde{\pi}_{k}, k=1,2,4,5$ and $\pi_{3}$ tend to the identity
matrix exponentially fast as $t\rightarrow+\infty$. The next task is to construct a global parametrix with constant jump on the segment $[-i,i]$ and two local parametrices near the critical points $z=\pm i$.
\subsection{Global parametrix on $[-i,i]$}
\begin{rhp} \label{RHP: Ptildeinfty}
We look for a $2 \times 2$ matrix-valued function $\widetilde{P}^{(\infty)}(z)$ satisfying the following properties.
\begin{itemize}
\item[\rm (1)] $\widetilde{P}^{(\infty)}(z)$ is analytic in $\mathbb{C} \setminus [-i,i]$.
\item[\rm (2)] $\widetilde{P}^{(\infty)}(z)$  satisfies the jump condition
  \begin{equation}\label{eq:Ptildejump}
  \widetilde{P}^{(\infty)}_{+}(z)=\widetilde{P}^{(\infty)}_{-}(z)
 \begin{pmatrix}
  0& 1 \\
  -1 & 0
  \end{pmatrix} ,\qquad z\in [-i,i].
  \end{equation}
  \item[\rm (3)] As $z\to \infty$, we have
  \begin{equation}\label{eq:PtildeInfty}
  \widetilde{P}^{(\infty)}(z)=I+O\left (\frac 1 {z}\right ).
 \end{equation}
\end{itemize}
\end{rhp}
A solution of the above RH problem is given by
\begin{equation}\label{Solu:PtildeInfty}
\widetilde{P}^{(\infty)}(z)=
\begin{pmatrix}1& 1 \\i& -i\end{pmatrix}w(z)^{\sigma_3}\begin{pmatrix}1& 1 \\i& -i\end{pmatrix}^{-1},\quad w(z)=\left(\frac{z-i}{z+i}\right)^{\frac{1}{4}},
\end{equation}
where the branches are taken such that $\arg(z\pm i)\in(-\frac{\pi}{2},\frac{3\pi}{2})$.
\subsubsection{Local parametrices near $\pm i$}
In this subsection, we seek two  parametrices $P^{(\pm i)}(z)$ satisfying the same jumps on the jump contours as $U(z)$ respectively in the neighborhoods $U(\pm i,\delta)$ of the saddle points $\pm i$, and matching with $\widetilde{P}^{(\infty)}(z)$ on the boundaries $\partial U(\pm i,\delta)$.
\begin{rhp} \label{RHP: Pilocal}
We look for a $2 \times 2$ matrix-valued function $P^{(i)}(z)$ satisfying the following properties.
\begin{itemize}
\item[\rm (1)] $P^{( i)}(z)$ is analytic in $U(i,\delta)\setminus \Sigma_{U}$.
\item[\rm (2)] $P^{( i)}(z)$  satisfies the same jump condition as $U(z)$ on $U(i,\delta)\cap\Sigma_{U}$.
  \item[\rm (3)] On $\partial U( i,\delta)$, we have
  \begin{equation}\label{eq:Pmatch}
  P^{(i)}(z)=\left(I+O(t^{-1})\right)\widetilde{P}^{(\infty)}(z),\quad \mathrm{as}\quad t\to \infty.
 \end{equation}
\end{itemize}
\end{rhp}
To construct a solution to the above RH problem, we first define a conformal mapping
\begin{equation}\label{lambda}
\lambda(z)=\left(\frac{3}{2}g(z)\right)^{\frac{2}{3}}=
e^{\frac{\pi}{2}i}2^{\frac{1}{3}}(z-i)(1+o(1)),\quad z\rightarrow i.
\end{equation}
Then, the solution to the above RH problem can be built out of the Airy function as follows:
\begin{equation}\label{Pi}
P^{(i)}(z)=E^{(i)}(z)\Phi^{\mathrm{(Ai)}}\left(t^{\frac{2}{3}}\lambda(z)\right)e^{\frac{1}{2}\pi i \nu\sigma_{3}}\sigma_{3}\varphi(z)^{-\nu\sigma_3}e^{-tg(z)\sigma_{3}},
\end{equation}
where $\Phi^{\mathrm{(Ai)}}$ denotes the Airy parametrix given in  Appendix \ref{AP}  and $E^{(i)}(z)$ is given by
\begin{equation}\label{Ei}
E^{(i)}(z)=\widetilde{P}^{(\infty)}(z)\varphi(z)^{\nu\sigma_{3}}\sigma_{3}e^{-\frac{1}{2}\pi i \nu\sigma_{3}}\begin{pmatrix}-i &i\\1&1\end{pmatrix}\left(t^{\frac{2}{3}}\lambda(z)\right)^{\frac{\sigma_{3}}{4}}.
\end{equation}
Here, the branch of the function $\lambda(z)^{\frac{1}{4}}$ is chosen such that $\arg\lambda(z)\in(0,2\pi)$. Particularly, we have on the segment $[-i,i]$
\begin{equation}\label{lambdajump}
\left(\lambda(z)^{\frac{\sigma_3}{4}}\right)_+=
\left(\lambda(z)^{\frac{\sigma_3}{4}}\right)_-e^{\frac{\pi i}{2}\sigma_3}.
\end{equation}
From \eqref{eq:Ptildejump} and \eqref{lambdajump}, it is straightforward to check that $E^{(i)}(z)$ is analytic in the neighborhood $U( i,\delta)$.
Furthermore, using \eqref{Solu:PtildeInfty}, \eqref{Pi} and \eqref{AiryAsyatinfty}, we arrive at the following  matching condition
\begin{equation}\label{JR:i}
\begin{aligned}
&P^{(i)}(z)\widetilde{P}^{(\infty)}(z)^{-1}\\
&=W^{(i)}(z)\left(I+\frac{1}{48\lambda(z)^{\frac{3}{2}}t}\begin{pmatrix}-1&6i\\6i&1\end{pmatrix}+O(t^{-2})\right)W^{(i)}(z)^{-1},\\
&=I+Q^{(i)}(z)t^{-1}+O\left(t^{-2}\right),
\end{aligned}
\end{equation}
where
\begin{equation}
W^{(i)}(z)=\widetilde{P}^{(\infty)}(z)\varphi(z)^{\nu\sigma_{3}}\sigma_{3}e^{-\frac{1}{2}\pi i\nu\sigma_{3}},
\end{equation}
and
\begin{equation}
Q^{(i)}(z)=\frac{1}{96\lambda(z)^{\frac{3}{2}}}\begin{pmatrix}c_{1}(z) &*\\ * & *\end{pmatrix},
\end{equation}
with
\begin{equation*}
c_{1}(z)=\left(-1-3\varphi(z)^{-2\nu}e^{\pi i\nu}-3\varphi(z)^{2\nu}e^{-\pi i\nu}\right)w(z)^{2}+\left(-1+3\varphi(z)^{-2\nu}e^{\pi i\nu}+3\varphi(z)^{2\nu}e^{-\pi i\nu}\right)w(z)^{-2} .\\
\end{equation*}
\begin{rhp} \label{RHP: P-ilocal}
We look for a $2 \times 2$ matrix-valued function $P^{(-i)}(z)$ satisfying the following properties.
\begin{itemize}
\item[\rm (1)] $P^{(-i)}(z)$ is analytic in $U(-i,\delta)\setminus \Sigma_{U}$.
\item[\rm (2)] $P^{(-i)}(z)$  satisfies the same jump condition as $U(z)$ on $U(-i,\delta)\cap\Sigma_{U}$.
  \item[\rm (3)] On $\partial U(-i,\delta)$, we have
  \begin{equation}\label{eq:P-match}
  P^{(-i)}(z)=\left(I+O(t^{-1})\right)\widetilde{P}^{(\infty)}(z),\quad \mathrm{as}\quad t\to \infty.
 \end{equation}
\end{itemize}
\end{rhp}
Similarly, the solution to the above RH problem can be built out of the Airy function as follows:
\begin{equation}\label{P-i}
P^{(-i)}(z)=\left\{\begin{aligned}
&E^{(-i)}(z)\Phi^{\mathrm{(Ai)}}\left(t^{\frac{2}{3}}\lambda(-z)\right)\sigma_{2}e^{\frac{3}{2}\pi i \nu\sigma_{3}}\varphi(z)^{-\nu\sigma_3}e^{-tg(z)\sigma_{3}},&  \Re z<0,\\
&E^{(-i)}(z)\Phi^{\mathrm{(Ai)}}\left(t^{\frac{2}{3}}\lambda(-z)\right)\sigma_{2}e^{-\frac{1}{2}\pi i \nu\sigma_{3}}\varphi(z)^{-\nu\sigma_3}e^{-tg(z)\sigma_{3}},&  \Re z>0,\\
\end{aligned}
\right.
\end{equation}
where $\Phi^{\mathrm{(Ai)}}$ is the Airy parametrix given in  Appendix \ref{AP}, $\lambda(z)$ is defined in \eqref{lambdajump} and $E^{(-i)}(z)$ is given by
\begin{equation}\label{E-i}
E^{(-i)}(z)=\left\{\begin{aligned}
&\widetilde{P}^{(\infty)}(z)\varphi(z)^{\nu\sigma_{3}}e^{-\frac{3}{2}\pi i \nu\sigma_{3}}\sigma_{2}\begin{pmatrix}-i &i\\1&1\end{pmatrix}
\left(t^{\frac{2}{3}}\lambda(-z)\right)^{\frac{\sigma_{3}}{4}},& \Re z<0,\\
&\widetilde{P}^{(\infty)}(z)\varphi(z)^{\nu\sigma_{3}}e^{\frac{1}{2}\pi i \nu\sigma_{3}}\sigma_{2}\begin{pmatrix}-i &i\\1&1\end{pmatrix}
\left(t^{\frac{2}{3}}\lambda(-z)\right)^{\frac{\sigma_{3}}{4}},& \Re z>0.\\
\end{aligned}
\right.
\end{equation}
It is readily seen that $E^{(-i)}(z)$ is analytic in $U(-i,\delta)$ and the function $P^{(-i)}(z)$ in \eqref{P-i}
satisfies the matching condition \eqref{eq:P-match}. It follows from \eqref{Solu:PtildeInfty}, \eqref{P-i} and \eqref{AiryAsyatinfty} that
\begin{equation}\label{JR:-i}
\begin{aligned}
&P^{(-i)}(z)\widetilde{P}^{(\infty)}(z)^{-1}\\
&=W^{(-i)}(z)\left(I+\frac{1}{48\lambda(-z)^{\frac{3}{2}}t}
\begin{pmatrix}-1&6i\\6i&1\end{pmatrix}+O\left(t^{-2}\right)\right)W^{(-i)}(z)^{-1},\\
&=I+Q^{(-i)}(z)t^{-1}+O\left(t^{-2}\right),
\end{aligned}
\end{equation}
where
\begin{equation}
W^{(-i)}(z)=\left\{\begin{aligned}
&\widetilde{P}^{(\infty)}(z)\varphi(z)^{\nu\sigma_{3}}e^{-\frac{3}{2}\pi i \nu\sigma_{3}}\sigma_{2},\quad \Re z<0,\\
&\widetilde{P}^{(\infty)}(z)\varphi(z)^{\nu\sigma_{3}}e^{\frac{1}{2}\pi i \nu\sigma_{3}}\sigma_{2},\quad \Re z>0,\\
\end{aligned}
\right.
\end{equation}
and
\begin{equation}
Q^{(-i)}(z)=\frac{1}{96\lambda(-z)^{\frac{3}{2}}}\begin{pmatrix}c_{2}(z) &*\\ * & *\end{pmatrix},
\end{equation}
with
\begin{equation*}
\begin{aligned}
c_{2}(z)=(1-3\varphi(z)^{-2\nu}e^{-\pi i\nu}-3\varphi(z)^{2\nu}e^{\pi i\nu})w(z)^{2}+(1+3\varphi(z)^{-2\nu}e^{-\pi i\nu}+3\varphi(z)^{2\nu}e^{\pi i\nu})w(z)^{-2}.\\
\end{aligned}
\end{equation*}

\subsection{Final transformation}
The final transformation is defined as
\begin{equation}\label{deformR}
\widetilde{R}(z)=\left\{\begin{aligned}
&U(z)P^{( \pm i)}(z)^{-1},\quad &z&\in U(\pm i,\delta)\setminus\Sigma_{U},\\
&U(z)\widetilde{P}^{(\infty)}(z)^{-1},&\mathrm{e}&\mathrm{lsewhere}.\\ \end{aligned}
\right.
\end{equation}

\begin{figure}[h]
  \centering
  \includegraphics[width=6.5cm,height=6.7cm]{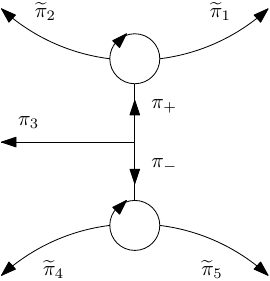}\\
  \caption{The jump contours $\Sigma_{\widetilde{R}}$ of the RH problem for $\widetilde{R}(z)$}\label{sigmaR}
\end{figure}
\begin{rhp} \label{RHP: finalR}
$\widetilde{R}(z)$ satisfies the following properties.
\item{(1)} $\widetilde{R}(z)$ is analytic for $z\in \mathbb{C}\setminus\Sigma_{\widetilde{R}}$, where the contour $\Sigma_{\widetilde{R}}$ is illustrated in Figure \ref{sigmaR}.
\item{(2)} For  $z\in \Sigma_{\widetilde{R}}$, we have $\widetilde{R}_+(z)=\widetilde{R}_-(z)J_{\widetilde{R}}(z)$, where
 \begin{equation}\label{JumpRtilde}
 J_{\widetilde{R}}(z)=\left\{\begin{aligned}
&P^{(\pm i)}(z)\widetilde{P}^{(\infty)}(z)^{-1},\quad &z&\in \partial U(\pm i,\delta),\\
&\widetilde{P}^{(\infty)}(z)J_{U}(z)\widetilde{P}^{(\infty)}(z)^{-1},\quad &\mathrm{e}&\mathrm{lsewhere}.
\end{aligned}
\right.
 \end{equation}
\item{(3)} As $z\rightarrow\infty$, we have
\begin{equation}\label{Rtilde:infty}
\widetilde{R}(z)=I+\frac{\widetilde{R}_1}{z}+O\left(\frac 1 {z^{2}}\right).
\end{equation}
\end{rhp}
Based on the matching conditions \eqref{eq:Pmatch} and \eqref{eq:P-match}, we have the following estimates as $t\rightarrow\infty$
\begin{equation}\label{JRtildeestimation}
J_{\widetilde{R}}(z)=\left\{\begin{aligned}
&I+O\left(t^{-1}\right), &z&\in \partial U(\pm i,\delta),\\
&I+O\left(e^{-c_{1}t}\right),&\mathrm{e}&\mathrm{lsewhere},
\end{aligned}\right.
\end{equation}
where $c_{1}$ is a positive constant. Consequently, we have the asymptotic approximation
\begin{equation}\label{Rtildeestimation}
\widetilde{R}(z)=I+O\left(t^{-1}\right),\quad \mathrm{as}\quad t\rightarrow \infty,
\end{equation}
where the error term is uniform for $z\in \mathbb{C}\setminus\Sigma_{\widetilde{R}}$.

\subsection{Asymptotics for $H$ and $u$ as $x\rightarrow+\infty$ }
According to our previous RH analysis $\widetilde{\Psi}\mapsto\Phi\mapsto U\mapsto \widetilde{R}$, for large $z$, we have
\begin{equation}\label{RPtilde}
2^{\nu\sigma_{3}}\Phi(z)e^{-tg(z)\sigma_{3}}\varphi(z)^{-\nu\sigma_{3}}=\widetilde{R}(z)\widetilde{P}^{(\infty)}(z).
\end{equation}
It follows from the large  $z$ asymptotics \eqref{rescaling+:infty} and \eqref{eq:g}  that
\begin{equation}\label{Phiasymp}
2^{\nu\sigma_{3}}\Phi(z)e^{-tg(z)\sigma_{3}}\varphi(z)^{-\nu\sigma_{3}}
=I+\frac{2^{\nu\sigma_{3}}\Phi_{1}2^{-\nu\sigma_{3}}-\frac{1}{8}t\sigma_{3}}{z}+
O\left(\frac{1}{z^2}\right).
\end{equation}
The asymptotic approximation of $\widetilde{P}^{(\infty)}(z)$ is  obtained by expanding  \eqref{Solu:PtildeInfty},
\begin{equation}\label{Ptilde:infty}
\widetilde{P}^{(\infty)}(z)=I+\frac{\widetilde{P}^{(\infty)}_1}{z}+
O\left(\frac{1}{z^2}\right),\quad \mathrm{as}~z\rightarrow\infty,
\end{equation}
where
\begin{equation}\label{Ptilde:infty1}
\widetilde{P}^{(\infty)}_1=\begin{pmatrix}
0 & -\frac{1}{2}\\
\frac{1}{2} & 0
\end{pmatrix}.
\end{equation}
From \eqref{Phiasymp} and the asymptotic behaviors \eqref{Rtilde:infty} and \eqref{Ptilde:infty}, we have
\begin{equation}\label{Phi1}
\Phi_{1}=2^{-\nu\sigma_{3}}\left(\widetilde{P}^{(\infty)}_1+\widetilde{R}_1\right)2^{\nu\sigma_{3}}+\frac{1}{8}t\sigma_{3},
\end{equation}
which implies
\begin{equation}\label{R12P12tilde}
(\Phi_{1})_{12}=4^{-\nu}\left((\widetilde{P}^{(\infty)}_{1})_{12}+(\widetilde{R}_{1})_{12}\right),
\end{equation}
and
\begin{equation}\label{R11P11tilde}
(\Phi_{1})_{11}=(\widetilde{P}^{(\infty)}_{1})_{11}+(\widetilde{R}_{1})_{11}+\frac{1}{8}t.
\end{equation}
In virtue of \eqref{JR:i}, \eqref{JR:-i} and \eqref{JumpRtilde}, we have as $ t\rightarrow\infty$
\begin{equation}\label{R1tilde:asym}
\widetilde{R}(z)=I+\frac{\widetilde{R}^{(1)}(z)}{t}+O\left(\frac 1 {t^{2}}\right),\quad z\in\partial U(\pm i,\delta),\end{equation}
where the coefficient $\widetilde{R}^{(1)}(z)=O(z^{-1})$ as $z\rightarrow\infty$, and
$\widetilde{R}^{(1)}(z)$  
satisfies the jump relation
\begin{equation}
\widetilde{R}^{(1)}_{+}(z)-\widetilde{R}^{(1)}_{-}(z)=Q^{(\pm i)}(z),\quad z\in\partial U(\pm i,\delta).
\end{equation}
The solution to the above RH problem is explicitly given by
\begin{equation}
\widetilde{R}^{(1)}(z)=\left\{\begin{aligned}
&\frac{C_{+}}{z-i}+\frac{C_{-}}{z+i}+\frac{D_{+}}{(z-i)^2}+\frac{D_{-}}{(z+i)^2}, &z \notin \overline{U(i,\delta)}\cup\overline{U(-i,\delta)},\\
&\frac{C_{+}}{z-i}+\frac{C_{-}}{z+i}+\frac{D_{+}}{(z-i)^2}+\frac{D_{-}}{(z+i)^2}-Q^{(\pm i)}(z), &z \in U(i,\delta)\cup U(-i,\delta),
\end{aligned}\right.
\end{equation}
where $C_{\pm}=\Res(Q^{(\pm i)}(z),z=\pm i)$.
Expanding this as $z\rightarrow\infty$, we obtain
\begin{equation}\label{R11tilde}
\widetilde{R}_{1}=\frac{\widetilde{R}^{(1)}_{1}}{t}+O\left(\frac 1 {t^{2}}\right), \quad t\rightarrow\infty,
\end{equation}
where $\widetilde{R}^{(1)}_{1}=C_{+}+C_{-}$,
after a direct computation by expanding $\varphi(z)$, $w(z)$ and $\lambda(\pm z)^{\frac{3}{2}}$ as $z\rightarrow\pm i$   , it follows that
\begin{equation}
(\widetilde{R}^{(1)}_{1})_{11}=(C_{+}+C_{-})_{11}=-\frac{1}{2}\nu^2+\frac{1}{8}.
\end{equation}
Substituting  \eqref{Ptilde:infty1}  and  \eqref{R11tilde}  into \eqref{R12P12tilde} and \eqref{R11P11tilde},  and
using the formula \eqref{eq:psi1} and \eqref{rescaling+}, we obtain \eqref{eq:uasyinfty} and \eqref{eq:Hasyinfty}.
\section{Asymptotics for $\Psi$ as $x\rightarrow-\infty$ }\label{sec:Psiinfty}
\subsection{Asymptotics for $\Psi$ as $x\rightarrow-\infty$ with $\nu \in\mathbb{C}\setminus \{0\}$} 
 First, we  introduce the following re-scaling of variable
 \begin{equation}\label{rescaling-}
 \widehat{\Phi}(z)=(-x)^{-\frac{\nu}{2}\sigma_3}\Psi\left((-x)^{\frac{1}{2}}z;x\right).
 \end{equation}
\begin{rhp}
$\widehat{\Phi}(z)$ solves the following RH problem.
\begin{itemize}
\item[\rm (1)] $ \widehat{\Phi}(z)$ is analytic in $\mathbb{C} \setminus\cup_{k=1}^5 ~\Gamma_k$; see Figure \ref{ModelRHP}.
\item[\rm (2)] $ \widehat{\Phi}(z)$  satisfies $ \widehat{\Phi}_{+}(z)= \widehat{\Phi}_{-}(z)J_\Psi$, $z\in\cup_{k=1}^5 ~\Gamma_k$ with $J_\Psi$ given in\eqref{eq:Sk}.
  \item[\rm (3)] As $z\to \infty$, we have
 \begin{equation}\label{hatPhi:infty}
\widehat{\Phi}(z)=\left(I+\frac{\widehat{\Phi}_{1}}{z}+O\left(\frac{1}{z^2}\right)\right)e^{t\widehat{\theta}(z)\sigma_3}z^{\nu\sigma_3},
\end{equation}
where $t=|x|^{\frac{3}{2}}$,  $\widehat{\theta}(z)=\frac{1}{3}z^3-\frac{1}{2}z$, and $\arg z \in(-\pi, \pi)$.
\end{itemize}
\end{rhp}
We introduce the $g$-function
\begin{equation}\label{hatg-function}
\widehat{g}(z)=\widehat{\theta}(z)-\widehat{\theta}\left(- {\sqrt{2}}/{2}\right).
 \end{equation}
It is direct to see that $\widehat{g}(z)$ has two saddle points at $z_{\pm}=\pm \frac{\sqrt{2}}{2}$.
The second transformation is defined by
\begin{equation}\label{eq:Uhat}
\widehat{U}(z)=e^{-t\widehat{\theta}\left(-\frac{\sqrt{2}}{2}\right)\sigma_{3}}
\widehat{\Phi}(z)e^{-t\widehat{g}(z)\sigma_{3}}.
\end{equation}
We also deform the jump contours $\Gamma_{k}$ into the steepest descent curves $\widehat{\Gamma}_{k}$, $k=1,2,4,5$  of the function $\widehat{g}(z)$. Then, we have
\begin{equation}\label{eq:hatgSign}
\Re \widehat{g}(z)<0, \quad z\in \widehat{\Gamma}_{1}\cup \widehat{\Gamma}_{3}\cup  \widehat{\Gamma}_{5}\cup  \widehat{\Gamma}_{6}, ~~\mbox{and}~~ \Re \widehat{g}(z)>0, \quad z\in \widehat{\Gamma}_{2}\cup \widehat{\Gamma}_{4},
\end{equation}
where the curves $\widehat{\Gamma}_{k}$, $k=1,2,3,4,5,6$ are  illustrated in   Figure \ref{NEWdeformation1}.
\begin{rhp}
$\widehat{U}(z)$ satisfies the following RH problem.
\begin{itemize}
\item[\rm (1)] $\widehat{U}(z)$ is analytic in $\mathbb{C} \setminus \Sigma_{\widehat{U}}$, where  $\Sigma_{\widehat{U}}$ is indicated in  Figure \ref{NEWdeformation1}.
\item[\rm (2)] $\widehat{U}(z)$  satisfies $\widehat{U}_{+}(z)=\widehat{U}_{-}(z)J_{\widehat{U}}(z)$, where
  \begin{equation}\label{eq:Uhatjump}
 J_{\widehat{U}}(z)=\left\{
\begin{aligned}
&\begin{pmatrix}1 & e^{2t\widehat{g}(z)-\pi i \nu}\\ 0 & 1 \end{pmatrix},\quad &z&\in\widehat{\Gamma}_{1},\\
&\begin{pmatrix}1 & 0\\  -e^{-2t\widehat{g}(z)+\pi i \nu} & 1 \end{pmatrix},\quad &z&\in\widehat{\Gamma}_{2},\\
&\begin{pmatrix}e^{-2\pi i \nu}& (e^{\pi i \nu}-e^{-\pi i \nu})e^{2t\widehat{g}(z)}\\0& e^{2\pi i \nu}\end{pmatrix},\quad &z&\in\widehat{\Gamma}_{3},\\
&\begin{pmatrix}1 & 0\\ e^{-2t\widehat{g}(z)-\pi i \nu} & 1 \end{pmatrix},\quad &z&\in\widehat{\Gamma}_{4},\\
&\begin{pmatrix}1 & -e^{2t\widehat{g}(z)+\pi i \nu}\\ 0 & 1 \end{pmatrix},\quad &z&\in\widehat{\Gamma}_{5},\\
&\begin{pmatrix}1& (e^{-\pi i \nu}-e^{\pi i \nu})e^{2t\widehat{g}(z)}\\ 0 & 1\end{pmatrix},\quad &z&\in\widehat{\Gamma}_{6}.
\end{aligned}
\right.
  \end{equation}
  \item[\rm (3)] As $z\to \infty$, we have
  \begin{equation}\label{eq:UhatInfty}
  \widehat{U}(z)=\left(I+\frac{\widehat{U}_1}{z}+O\left(\frac 1 {z^2}\right)\right)z^{\nu\sigma_3},
 \end{equation}
 where  $\arg z \in(-\pi, \pi)$.
\end{itemize}
\end{rhp}

\begin{figure}[h]
  \centering
  \includegraphics[width=8.7cm,height=5.6cm]{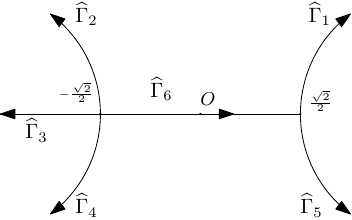}\\
  \caption{ The jump contours $\Sigma_{\widehat{U}}$ of RH problem for $\widehat{U}(z)$}\label{NEWdeformation1}
\end{figure}

It is seen from \eqref{eq:hatgSign} that the jump matrices on $\widehat{\Gamma}_{k},  k=1,2,4,5,6$  tend to the identity
matrix exponentially fast as $t\rightarrow+\infty$. While, the jump on $\widehat{\Gamma}_{3}$ tends to $e^{-2\pi i \nu\sigma_3}$  exponentially fast as $t\rightarrow+\infty$. In the following two sections, we construct a global parametrix with the remaining constant jump $e^{-2\pi i \nu\sigma_3}$  on  $\widehat{\Gamma}_{3}$  and a local parametrix near the saddle points $z=-\frac{\sqrt{2}}{2}$.
\subsubsection{Global parametrix}

\begin{rhp} \label{RHP: Pinftyhat}
We look for a $2 \times 2$ matrix-valued function $\widehat{P}^{(\infty)}(z)$ satisfying the following properties.
\begin{itemize}
\item[\rm (1)] $\widehat{P}^{(\infty)}(z)$ is analytic in $\mathbb{C} \setminus (-\infty,z_{-}]$.
\item[\rm (2)] $\widehat{P}^{(\infty)}(z)$  satisfies the jump condition
  \begin{equation}\label{eq:Phatjump}
  \widehat{P}^{(\infty)}_{+}(z)=\widehat{P}^{(\infty)}_{-}(z)J_{\infty}(z)
  \end{equation}
  where
  $$
  J_{\infty}(z)=\begin{aligned}
 &\begin{pmatrix}e^{2\pi i \nu}& 0 \\0 & e^{-2\pi i \nu}\end{pmatrix} ,\qquad z\in (-\infty,z_{-}]\\
\end{aligned}.
  $$
  \item[\rm (3)] As $z\to \infty$, we have
  \begin{equation}\label{eq:PhatInfty}
  \widehat{P}^{(\infty)}(z)=\left(I+O\left(\frac 1 {z}\right)\right)z^{\nu\sigma_{3}},
 \end{equation}
 where $\arg z \in(-\pi,\pi)$.
\end{itemize}
\end{rhp}
The solution of the above RH problem is explicitly given by
\begin{equation}\label{Solu:Pinftyhat}
\widehat{P}^{(\infty)}(z)=\left(z+\frac{\sqrt{2}}{2}\right)^{\nu \sigma_{3}},
\end{equation}
where $\arg \left (z+ \frac{\sqrt{2}} {2}\,\right)  \in(-\pi, \pi)$.
\subsubsection{Local parametrices near $z_{-}=-\frac{\sqrt{2}}{2}$}
In this subsection, we seek a function $ P^{(z_{-})}(z)$ satisfying the same jumps as $\widehat{U}(z)$ on the jump contours $\Sigma_{\widehat{U}}$ in the neighborhood $U(z_{-},\delta)$ of the saddle points $z_{-}=-\frac{\sqrt{2}}{2}$ and matching with $\widehat{P}^{(\infty)}(z)$ on the boundary $\partial U(z_{-},\delta)$.
\begin{rhp} \label{RHP: P-local}
We look for a $2 \times 2$ matrix-valued function $P^{(z_{-})}(z)$ satisfying the following properties.
\begin{itemize}
\item[\rm (1)] $ P^{(z_{-})}(z)$ is analytic in $U(z_{-},\delta)\setminus \Sigma_{\widehat{U}}$.
\item[\rm (2)] $ P^{(z_{-})}(z)$  satisfies the same jump condition as $\widehat{U}(z)$ on $U(z_{-},\delta)\cap\Sigma_{\widehat{U}}$.
  \item[\rm (3)] On $\partial U(z_{-},\delta)$, we have
  \begin{equation}\label{eq:P-Infty}
 P^{(z_{-})}(z)=t^{-\frac{\nu}{2}\sigma_{3}}
 \left(I+O\left(t^{-\frac{1}{2}}\right)\right)
 t^{\frac{\nu}{2}\sigma_{3}}\widehat{P}^{(\infty)}(z),\quad \mathrm{as}\quad t\to \infty.
 \end{equation}
\end{itemize}
\end{rhp}
In the case $\nu \in \mathbb{C}\setminus \mathbb{N}$, we  define the following conformal mapping
\begin{equation}\label{eta}
\eta(z)=2\sqrt{-\widehat{g}(z)}=2^{\frac{3}{4}}\left(z+\frac{\sqrt{2}}{2}\right)
\left(1+o(1)\right),\quad z\rightarrow -\frac{\sqrt{2}}{2}.
\end{equation}
Let $\Phi^{(PC)}$ be the parabolic cylinder parametrix given in Appendix \ref{PCP}. Then the parametrix near
$z=z_{-}$ can be constructed as follows:
\begin{equation}\label{P-}
P^{(z_{-})}(z)=E^{(z_{-})}(z)\Phi^{(PC)}\left(t^{\frac{1}{2}}\eta(z)\right)
\left(-e^{-\pi i \nu}h_{1}\right)^{\frac{\sigma_3}{2}}D(z)\sigma_{1}e^{-t\widehat{g}(z)\sigma_{3}},
\end{equation}
where $ h_{1}=\frac{\sqrt{2 \pi}}{\Gamma(-\nu)} e^{i \pi \nu}$ for $\nu \in \mathbb{C}\setminus \mathbb{N}$, as defined in \eqref{h0}, and $E^{(z_{-})}(z)$ is defined as
\begin{equation}\label{Ehat}
E^{(z_{-})}(z)=\widehat{P}^{(\infty)}(z)\sigma_{1}\eta(z)^{\nu\sigma_{3}}t^{\frac{1}{2}\nu \sigma_{3}}\left(-e^{-\pi i \nu}h_{1}\right)^{-\frac{\sigma_3}{2}} 2^{-\frac{\sigma_3}{2}}\begin{pmatrix}t^{\frac{1}{2}}\eta(z) & 1 \\ 1 & 0\end{pmatrix},
\end{equation}
with
\begin{equation}\label{eq:D}
D(z)=\left\{\begin{aligned}
&I, \quad &\arg z&\in(-\frac{\pi}{4},\pi),\\
&e^{2\pi i \nu\sigma_{3}}, \quad &\arg z&\in(\pi,\frac{7\pi}{4}).
\end{aligned}\right.
\end{equation}
Here, the branch of the function $\eta(z)^{\nu}$ is chosen such that $\arg\eta(z)\in(-\pi,\pi)$.
It is easily seen that $E^{(z_{-})}(z)$ is analytic in the neighborhood $U(z_{-},\delta)$. Furthermore, a combination of
\eqref{Solu:Pinftyhat}, \eqref{P-}, \eqref{Ehat} and \eqref{PCAsyatinfty} gives
\begin{equation}\label{P-:match}
\begin{aligned}
&P^{(z_{-})}(z)\widehat{P}^{(\infty)}(z)^{-1}\\
&=W^{(z_{-})}(z)\begin{pmatrix}1+\frac{\nu(\nu+1)}{2t\eta(z)^{2}}+O(t^{-2})&\frac{\nu}{t^{\frac{1}{2}}\eta(z)}+O(t^{-\frac{3}{2}})\\
\frac{1}{t^{\frac{1}{2}}\eta(z)}+O(t^{-\frac{3}{2}})&1-\frac{\nu(\nu-1)}{2t\eta(z)^{2}}+O(t^{-2})\end{pmatrix}W^{(z_{-})}(z)^{-1},\\
&=I+Q^{(z_{-})}(z)t^{-\frac{1}{2}}+O(t^{-1}),
\end{aligned}
\end{equation}
where
\begin{equation}
W^{(z_{-})}(z)=\widehat{P}^{(\infty)}(z)\sigma_{1}\eta(z)^{\nu\sigma_{3}}
t^{\frac{1}{2}\nu\sigma_{3}}\left(-e^{-\pi i \nu}h_{1}\right)^{-\frac{\sigma_{3}}{2}},
\end{equation}
and $Q^{(z_{-})}(z)$ is given by
\begin{equation}\label{eq:Q}
Q^{(z_{-})}(z)=t^{-\frac{1}{2}\nu\sigma_{3}}\begin{pmatrix}0&
-\frac{\left(z+\frac{\sqrt{2}}{2}\right)^{2\nu}h_{1}e^{-\pi i\nu}}{\eta(z)^{2\nu+1}}\\
-\frac{\nu\eta(z)^{2\nu-1} e^{\pi i\nu}}{h_{1}\left(z+\frac{\sqrt{2}}{2}\right)^{2\nu}}&0\end{pmatrix}
t^{\frac{1}{2}\nu\sigma_{3}}.
\end{equation}

\begin{figure}[h]
  \centering
  \includegraphics[width=8.5cm,height=6.8cm]{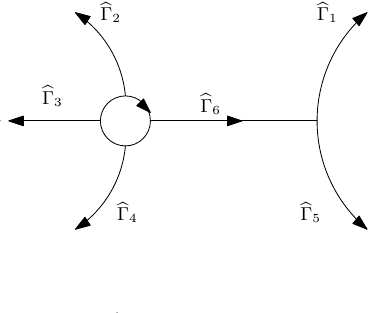}\\
  \caption{The jump contours $\Sigma_{\widehat{R}}$ of the RH problem for $\widehat{R}(z)$}\label{NEWfinalR}
\end{figure}
\subsubsection{Final transformation}
The final transformation is defined as
\begin{equation}\label{model:hatR}
\widehat{R}(z)=\left\{\begin{aligned}
&t^{\frac{1}{2}\nu\sigma_{3}}\widehat{U}(z)P^{(z_{-})}(z)^{-1}
t^{-\frac{1}{2}\nu\sigma_{3}},\quad &z&\in U(z_{-},\delta)\setminus\Sigma_{\widehat{U}},\\
&t^{\frac{1}{2}\nu\sigma_{3}}\widehat{U}(z)\widehat{P}^{(\infty)}(z)^{-1}t^{-\frac{1}{2}\nu\sigma_{3}},&\mathrm{e}&\mathrm{lsewhere}.\\ \end{aligned}
\right.
\end{equation}
Then $\widehat{R}(z)$ satisfies the following RH problem.
\begin{rhp} \label{RHP: hatR}
\item{(1)} $\widehat{R}(z)$ is analytic for $z\in \mathbb{C}\setminus\Sigma_{\widehat{R}}$, where the contour $\Sigma_{\widehat{R}}$ is illustrated in Figure \ref{NEWfinalR}.
\item{(2)} On the contour $\Sigma_{\widehat{R}}$, we have $\widehat{R}_+(z)=\widehat{R}_-(z)J_{\widehat{R}}(z)$, where
 \begin{equation}\label{JumpRhat}
 J_{\widehat{R}}(z)=\left\{\begin{aligned}
&t^{\frac{1}{2}\nu\sigma_{3}}P^{(z_{-})}(z)\widehat{P}^{(\infty)}(z)^{-1}t^{-\frac{1}{2}\nu\sigma_{3}},\quad &z&\in \partial U(z_{-},\delta),\\
&t^{\frac{1}{2}\nu\sigma_{3}}\widehat{P}^{(\infty)}(z)J_{\widehat{U}}(z)\widehat{P}^{(\infty)}(z)^{-1}t^{-\frac{1}{2}\nu\sigma_{3}},\quad &\mathrm{e}&\mathrm{lsewhere}.\\ \end{aligned}
\right.
 \end{equation}
\item{(3)} As $z\rightarrow\infty$, we have
\begin{equation}\label{Rhat:infty}
\widehat{R}(z)=I+\frac{\widehat{R}_1}{z}+O\left(\frac{1}{z^{2}}\right).
\end{equation}
\end{rhp}
By virtue of the matching condition \eqref{eq:P-Infty}, we have the following estimates
\begin{equation}\label{JRhatestimation}
J_{\widehat{R}}(z)=\left\{\begin{aligned}
&I+O(t^{-\frac{1}{2}}), &z&\in\partial U(z_{-},\delta),\\
&I+O(e^{-c_2t}),\quad &\mathrm{e}&\mathrm{lsewhere},
\end{aligned}\right.
\end{equation}
where $c_2$ is a positive constant. Consequently, we have
\begin{equation}\label{Rhat:estimation}
\widehat{R}(z)=I+O(t^{-\frac{1}{2}}), \quad \text{as} \quad t\rightarrow\infty,
\end{equation}
uniformly for $z\in \mathbb{C}\setminus\Sigma_{\widehat{R}}$. This completes the RH analysis for the case $\nu\in\mathbb{C}\setminus\mathbb{N}$.

For $\nu\in\mathbb{N}\setminus\{0\}$,  the asymptotics of $\Psi(z,x;\nu)$  can be obtained from that of  $\Psi(z,x;-\nu)$  by using the symmetry relation
\begin{equation}\label{eq: Symm}
\Psi(z,x;\nu) =(-1)^{\nu} \sigma_1 \Psi(-z,x;-\nu) \sigma_1,\quad \nu\in\mathbb{N}.
\end{equation}
\subsection{Asymptotics of $u(x)$ and $H(x)$ as $x\to-\infty$ for $\nu\in\mathbb{C}\setminus\{0\}$}
First, we   consider the case   $\nu\in\mathbb{C}\setminus\mathbb{N}$.
According to our previous RH analysis $\Psi\mapsto\widehat{\Phi}\mapsto \widehat{U}\mapsto \widehat{R}$, we have that, for large $z$,
\begin{equation}\label{RPhat}
e^{-t\widehat{\theta}(-\frac{\sqrt{2}}{2})\sigma_{3}}\widehat{\Phi}(z)e^{-t\widehat{g}(z)\sigma_{3}}=t^{-\frac{1}{2}\nu\sigma_{3}}\widehat{R}(z)t^{\frac{1}{2}\nu\sigma_{3}}\widehat{P}^{(\infty)}(z).
\end{equation}
The asymptotic approximation of $\widehat{P}^{(\infty)}(z)$ can be  obtained by expanding the explicit representation \eqref{Solu:Pinftyhat}.
\begin{equation}\label{Phat:infty}
\widehat{P}^{(\infty)}(z)=\left(I+\frac{\widehat{P}^{(\infty)}_1}{z}+O\left(\frac{1}{z^2}\right)\right)z^{\nu \sigma_{3}},\quad \mathrm{as} \quad z\rightarrow\infty,
\end{equation}
where
\begin{equation}\label{Phat:infty1}
\widehat{P}^{(\infty)}_1=\begin{pmatrix}
\frac{\sqrt{2}}{2}\nu & 0\\
0 & -\frac{\sqrt{2}}{2}\nu
\end{pmatrix}.
\end{equation}
From \eqref{RPhat} and the asymptotic behaviors \eqref{hatPhi:infty}, \eqref{Rhat:infty},  \eqref{Phat:infty} and \eqref{Phat:infty1}, we have
\begin{equation}\label{R12P12hat}
e^{-\frac{\sqrt{2}}{3}t}(\widehat{\Phi}_{1})_{12}=t^{-\nu}(\widehat{R}_{1})_{12}, \quad  e^{\frac{\sqrt{2}}{3}t}(\widehat{\Phi}_{1})_{21}=t^{\nu}(\widehat{R}_{1})_{21},
\end{equation}
and
\begin{equation}\label{R11P11hat}
(\widehat{\Phi}_{1})_{11}=(\widehat{P}^{(\infty)}_{1})_{11}+(\widehat{R}_{1})_{11}.
\end{equation}
In virtue of \eqref{P-:match}, \eqref{JumpRhat} and \eqref{Rhat:infty}, we easily get
\begin{equation}
\widehat{R}(z)=I+\frac{\widehat{R}^{(1)}(z)}{t^{ {1}/{2}}}+O\left(\frac 1 t\right), \quad\text{for}\quad z\in\partial U(z_{-},\delta),
\end{equation}
where $\widehat{R}^{(1)}(z)=O(z^{-1})$ as $z\rightarrow\infty$ and $\widehat{R}^{(1)}(z)$  satisfies the jump relation
\begin{equation}
\widehat{R}^{(1)}_{+}(z)-\widehat{R}^{(1)}_{-}(z)=Q^{(z_{-})}(z), \quad z\in\partial U(z_{-},\delta).
\end{equation}
The solution to the above RH problem can be constructed explicitly as follows:
\begin{equation}
\widehat{R}^{(1)}(z)=\left\{\begin{aligned}
&\frac{L_{-}}{z+\frac{\sqrt{2}}{2}}, &z&\notin \overline{U(z_{-},\delta)},\\
&\frac{L_{-}}{z+\frac{\sqrt{2}}{2}}-Q^{(z_{-})}(z), &z&\in  U(z_{-},\delta),
\end{aligned}\right.
\end{equation}
where $L_{-}=\Res\left(Q^{(z_{-})}(z),z=-\frac{\sqrt{2}}{2}\right)$. Expanding this as $z\rightarrow\infty$, we obtain the asymptotics for   $\widehat{R}_{1}$ in the expansion  \eqref{Rhat:infty}
\begin{equation}\label{R1hat:asym}
\widehat{R}_1=\frac{\widehat{R}^{(1)}_1}{t^{ {1}/{2}}}+O\left(\frac 1 t\right),\quad \mathrm{as} \quad t\rightarrow\infty,
\end{equation}
where
\begin{equation}
\widehat{R}^{(1)}_1=L_{-}=\begin{pmatrix}0&-2^{-\frac{3}{2}\nu-\frac{3}{4}}h_{1}e^{-\pi i\nu}\\
-\nu2^{\frac{3}{2}\nu-\frac{3}{4}}e^{\pi i\nu}h^{-1}_{1}&0\end{pmatrix}.
\end{equation}

Substituting  \eqref{R1hat:asym} and \eqref{Phat:infty1} into \eqref{R12P12hat} and \eqref{R11P11hat},  and
using 
 \eqref{eq:psi1} and \eqref{rescaling-}, we obtain for  $\nu\in \mathbb{C}\setminus \mathbb{N}$
\begin{equation}\label{u:asy}
u(x;\nu)=\sqrt{-\frac{x}{2}}+\frac{\nu+\frac{1}{2}}{2x}+O\left((-x)^{-\frac{5}{2}}\right),
\quad x\rightarrow -\infty,
\end{equation}
and
\begin{equation}\label{H:asy}
H(x;\nu)=-\nu\sqrt{-\frac{x}{2}}+O\left((-x)^{-1}\right),\quad x\rightarrow -\infty.
\end{equation}

Next, we consider the case $\nu\in\mathbb{N}\setminus\{0\}$.
From  \eqref{eq:psi1} and \eqref{eq: Symm}, we have
\begin{equation}\label{eq: Symmu}
u(x;\nu) =-\frac{d}{dx}\log (\Psi_1)_{21}(x;-\nu),
\end{equation}
and
\begin{equation}\label{eq: SymmH}
H(x;\nu) =H(x;-\nu)
\end{equation}
for $\nu\in\mathbb{N}\setminus\{0\}$.
Therefore,  we obtain from   \eqref{R12P12hat}, \eqref{H:asy}, \eqref{eq: Symmu} and \eqref{eq: SymmH}
that
\begin{equation}\label{u-:asy}
u(x;\nu)=-\sqrt{-\frac{x}{2}}+\frac{\nu+\frac{1}{2}}{2x}+
O\left((-x)^{-\frac{5}{2}}\right),\quad x\rightarrow -\infty,
\end{equation}
and
\begin{equation}\label{H-:asy}
H(x;\nu)=\nu\sqrt{-\frac{x}{2}}+O\left((-x)^{-1}\right),\quad x\rightarrow -\infty,
\end{equation}
for $\nu\in \mathbb{N}\setminus\{0\}$.

\subsection{Asymptotics for $\Psi$ as $x\rightarrow-\infty$ for $\nu=0$}
In the section, we consider the remaining case $\nu=0$. Still, we carry out a nonlinear steepest descent analysis of the Riemann-Hilbert problem for $\Psi$.
The first  transformation  is the same as \eqref{rescaling-}.
While  the second transformation is now defined by
\begin{equation}\label{tildeU}
\widetilde{U}(z)=\widehat{\Phi}(z)e^{-t\widehat{\theta}(z)\sigma_{3}},\quad\mbox{with}\quad \widehat{\theta}(z)=\frac{1}{3}z^3-\frac{1}{2}z .
\end{equation}
We also deform the original contour $\Gamma_{k},~k=1, 2, 4, 5$ into the anti-Stokes curves of $\widehat{\theta}(z)$ as shown in Figure \ref{deformation}.
Then $\widetilde{U}(z)$ solves the following RH problem.
\begin{figure}[ht]
  \centering
  \includegraphics[width=6.5cm,height=5.5cm]{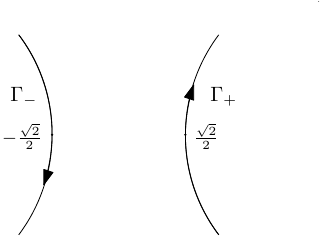}\\
  \caption{The jump curves $\Sigma_{\widetilde{U}}$ of the RH problem for $\widetilde{U}(z)$}\label{deformation}
\end{figure}

\begin{rhp}\label{RHP: tildeU}
\item{(1)} $\widetilde{U}(z)$ is analytic in $\mathbb{C}\setminus \Sigma_{\widetilde{U}}$.
\item{(2)} For $z\in \Sigma_{\widetilde{U}}$, we have $\widetilde{U}_{+}(z)=\widetilde{U}_{-}(z)J_{\widetilde{U}}(z)$, where
  \begin{equation}\label{eq:tildeUjump}
 J_{\widetilde{U}}(z)=\left\{
\begin{aligned}
&\begin{pmatrix}1 & e^{2t\widehat{\theta}(z)}\\ 0 & 1 \end{pmatrix},\quad &z&\in\Gamma_{+},\\
&\begin{pmatrix}1 & 0\\ e^{-2t\widehat{\theta}(z)} & 1 \end{pmatrix},\quad &z&\in\Gamma_{-}.
\end{aligned}
\right.
  \end{equation}
  \item{(3)} As $z\to \infty$, we have
  \begin{equation}\label{eq:tildeUInfty}
  \widetilde{U}(z)=\left (I+\frac{\widetilde{U}_1}{z}+O\left (\frac 1 {z^2}\right )\right ).
 \end{equation}
\end{rhp}
From \eqref{eq:tildeUjump}, we have
  \begin{equation}\label{eq:UjumpEst} J_{\widetilde{U}}(z)-I=O\left(e^{-2t\widehat{\theta}(z)}\right),
 \end{equation}
 where  $z\in \Sigma_{\widetilde{U}}$.
Therefore, we have the estimate
  \begin{equation}\label{eq:tildeUEst}||J_{\widetilde{U}}(z)-I||_{L^2\cap L^{\infty}(\Gamma_{+}\cup\Gamma_{-1})}=O\left(e^{-\frac{\sqrt{2}}{3}t}\right).
 \end{equation}
This ensures the  solvability  of the RH problem for $\widetilde{U}(z)$, and the solution can be given by the integral formula
\begin{equation}\label{U:solu}
\widetilde{U}(z)=I+\frac{1}{2\pi i}\int_{\Gamma_{+}\cup \Gamma_{-}}\frac{ \widetilde{U}_{-}(s)(J_{\widetilde{U}}(s)-I)}{s-z}ds.
\end{equation}
This expression  
 implies the   large-$z$ expansion
\begin{equation}\label{eq:U1est}
\widetilde{U}(z)=I+\frac{\widetilde{U}_1(s)}{z} +O\left(\frac 1{z^2}\right ),
\end{equation}
where

\begin{equation}\label{Res:U}
\widetilde{U}_{1}(t)=-\frac{1}{2\pi i}\int_{\Gamma_{+}\cup \Gamma_{-}}\widetilde{U}_{-}(s)\left( J_{\widetilde{U}}(s)-I\right)ds.
\end{equation}
To continue, we compute by using  \eqref{eq:tildeUjump} and \eqref{eq:tildeUEst}
\begin{equation}\label{U112}
\left(\widetilde{U}_{1}\right)_{12}=-\frac{1}{2\pi i }\int_{\Gamma_{+}\cup \Gamma_{-}} \left(J_{\widetilde{U}}(s)-I\right)_{12} ds+O\left(e^{-\frac{2\sqrt{2}}{3}t}\right)
=c\,e^{-\frac{\sqrt{2}}{3}t}t^{-\frac{1}{2}}\left(1+O\left(t^{-1}\right)\right),
\end{equation}
for some nonzero constant  $c$. We also have from  \eqref{eq:tildeUjump}
\begin{align}\label{U111}
\left(\widetilde{U}_{1}\right)_{11}&=-\frac{1}{2\pi i }\int_{\Gamma_{+}\cup \Gamma_{-}}  \left((\widetilde{U}_{-}(s)-I)(J_{\widetilde{U}}(s)-I)\right)_{11} ds\\
&= -\frac{1}{2\pi i }\int_{ \Gamma_{-}}  (\widetilde{U}_{-}(s)-I)_{12}e^{-t\widehat{\theta}(s)}ds\\
&=O\left(e^{-\frac{2\sqrt{2}}{3}t}\right).
\end{align}
Recalling the relation $\Psi_{1}=(-x)^{\frac{1}{2}}\widetilde{U}_{1}$,  we obtain from  \eqref{eq:psi1} and \eqref{U112} that
\begin{equation}\label{final:solu}
u(x;0)=-\sqrt{-\frac{x}{2}}+\frac{1}{4x}+O\left((-x)^{-\frac{5}{2}}\right),\quad x\rightarrow-\infty.
\end{equation}
From \eqref{eq:psi1} and \eqref{U111}, we have
\begin{equation}\label{final:solu-H}
H(x;0)=O\left((-x)^{1/2}e^{-\frac{2\sqrt{2}}{3}(-x)^{3/2}}\right),\quad x\rightarrow-\infty.
\end{equation}
Now combining \eqref{u:asy}-\eqref{H-:asy}, \eqref{final:solu} and \eqref{final:solu-H}, we further have \eqref{eq:uasyneginfty} and \eqref{eq:Hasyneginfty}.


\section{Proof of Theorem \ref{thm2} }\label{sectionDh}  
Tracing back the series of  invertible transformations $Y\to \widehat{Y}\to T\to S\to R$, we obtain
\begin{equation}\label{eq:YApprox1}
 Y(z)=n^{-\frac{ \nu}{3}\sigma_3} R(z)n^{\frac{ \nu}{3}\sigma_3} \left(\frac{z+1}{z}\right)^{\nu \sigma_3}e^{-\frac{t}{2z}\sigma_3}z^{n\sigma_3}\quad
 \mbox{for}\quad |z|>1+\delta, \end{equation}
 and
 \begin{equation}\label{eq:YApprox2}
 Y(z)=n^{-\frac{ \nu}{3}\sigma_3} R(z)n^{\frac{ \nu}{3}\sigma_3} (z+1)^{\nu \sigma_3}i\sigma_2e^{-\frac{t}{2}z\sigma_3}\quad
 \mbox{for}\quad  |z|<\delta ,
 \end{equation}
 where
  \begin{equation*}
   \sigma_2=\begin{pmatrix}0 & -i\\ i & 0\end{pmatrix}
  \end{equation*}
  is one of the Pauli matrices.

Thus, we get from \eqref{eq:R1O} and \eqref{eq:YApprox1} that
\begin{equation}\label{eq:Y1}
 (Y_{-1})_{11}= -\frac{t}{2}+\nu -\frac{1}{n^{1/3}}2^{\frac{2}{3}}H\left(2^{\frac{2}{3}}n^{2/3}(\tau-1);\nu\right)+
 O\left(\frac 1 {n^{ 2/3}}\right),
 \end{equation}
where the Hamiltonian $H(s;\nu)=-(\Psi_1)_{11}(s)$; cf.  \eqref{eq:psi1}.
Similarly,  using \eqref{eq:R1O} and \eqref{eq:YApprox2}, we obtain
\begin{align}
\frac{d}{dz} \log (Y)_{21}(z)|_{z=0}&= -\frac{t}{2}-\nu +\frac{d}{dz} \log (R)_{22}(z)|_{z=0}\nonumber\\
&= -\frac{t}{2}-\nu - 2^{\frac{2}{3}}n^{-\frac 1 3}
H\left(2^{\frac{2}{3}}n^{\frac 2 3}(\tau-1);\nu\right)+O\left(n^{-\frac 2 3}\right).\label{eq:Y0}
\end{align}
From the differential identity \eqref{eq:dD}, \eqref{eq:Y1} and \eqref{eq:Y0}, we have \eqref{eq:Dasy}.
In virtue of  \eqref{eq:hY}, \eqref{eq:R1O} and \eqref{eq:YApprox2}, we have
\begin{equation}\label{eq:hnasy}
h_{n}=Y_{12}(0)=1-2^{\frac{2}{3}}n^{-\frac 1 3}H\left(2^{\frac{2}{3}}n^{\frac 2 3}(\tau-1); \nu \right)+O\left(n^{-\frac 2 3}\right),
\end{equation}
\begin{equation}\label{eq:recuasy1}
\pi_{n}(0)=Y_{11}(0)=-2^{\frac{2}{3}+\frac{4}{3}\nu}n^{-\frac{1}{3}-\frac{2}{3}\nu}
e^{n\pi i}\left((\Psi_1)_{12}\left(2^{\frac{2}{3}}n^{\frac{2}{3}}(\tau-1);\nu\right)+O\left(n^{-\frac 1 3}\right)\right).
\end{equation}
Recalling \eqref{eq:psi1}, we obtain \eqref{eq:hasy} and \eqref{eq:Rasy}.
This completes the proof of Theorems \ref{thm2}.

\section*{Acknowledgements}
The authors are grateful to the reviewers for their constructive comments and suggestions.
The work of Shuai-Xia Xu was supported in part by the National Natural Science Foundation of China under grant numbers 11971492 and 12371257, and by  Guangdong Basic and Applied Basic Research Foundation (Grant No. 2022B1515020063). Yu-Qiu Zhao was supported in part by the National Natural Science Foundation of China under grant numbers 11971489 and 12371077.

\begin{appendices}

\section{Local parametrix models}
\subsection{Airy parametrix}\label{AP}
Define
\begin{equation}
\Phi^{(\mathrm{Ai})}(z)=M\left\{
\begin{aligned}
&\begin{pmatrix}
\mathrm{Ai}(e^{-\frac{2}{3}\pi i}z) & \mathrm{Ai}(z) \\
e^{-\frac{2}{3}\pi i} \mathrm{Ai}^{\prime}(e^{-\frac{2}{3}\pi i}z) & \mathrm{Ai}^{\prime}(z)
\end{pmatrix} e^{i \frac{\pi}{6} \sigma_{3}}, &z&\in \mathrm{I},\\
&\begin{pmatrix}
\mathrm{Ai}(e^{-\frac{2}{3}\pi i}z) & \mathrm{Ai}(z) \\
e^{-\frac{2}{3}\pi i} \mathrm{Ai}^{\prime}(e^{-\frac{2}{3}\pi i}z) & \mathrm{Ai}^{\prime}(z)
\end{pmatrix} e^{i \frac{\pi}{6} \sigma_{3}}\begin{pmatrix}
1 & -1\\
0 & 1
\end{pmatrix},&z&\in \mathrm{II},\\
&\begin{pmatrix}
\mathrm{Ai}(e^{-\frac{2}{3}\pi i}z) & \mathrm{Ai}(z) \\
e^{-\frac{2}{3}\pi i} \mathrm{Ai}^{\prime}(e^{-\frac{2}{3}\pi i}z) & \mathrm{Ai}^{\prime}(z)
\end{pmatrix} e^{i \frac{\pi}{6} \sigma_{3}}\begin{pmatrix}
0 & -1 \\
1 & 1
\end{pmatrix},&z&\in \mathrm{III},
\end{aligned}
\right.
\end{equation}
where $\mathrm{Ai}(z)$ is the Airy function, 
$$
M=\sqrt{2\pi}\begin{pmatrix}e^{\frac{1}{6}\pi i} & 0\\ 0 & -e^{\frac{1}{6}\pi i}\end{pmatrix},
$$
and the regions I-III are shown in Figure \ref{Airy}. Then, $\Phi^{(\mathrm{Ai})}(z)$ satisfies the following RH problem \cite[Section 3.7]{B}, which is related to the standard one in \cite[Chapter 7]{D1} and \cite{DKMVZ2} after some minor modifications.
\begin{figure}[h]
  \centering
  \includegraphics[width=6.5cm,height=5cm]{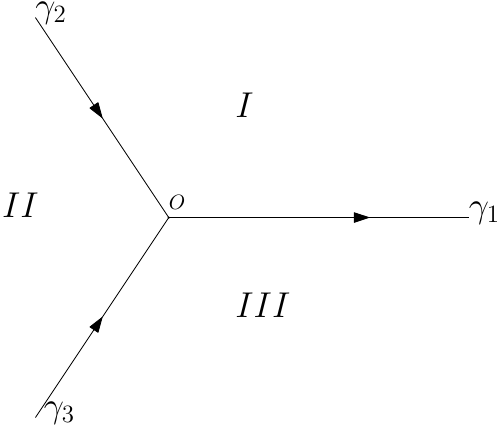}\\
  \caption{The jump contours and regions for $\Phi^{(\mathrm{Ai})}$}\label{Airy}
\end{figure}

\begin{rhp}
\item[\rm (1)] $\Phi^\mathrm{(Ai)}(z)$ is analytic for $s\in \mathbb{C}\setminus \cup^3_{k=1}\gamma_{k}$, see Figure \ref{Airy}.
\item[\rm (2)] $\Phi^\mathrm{(Ai)}(z)$ satisfies the jump relations $\Phi^\mathrm{(Ai)}_+(z)=\Phi^\mathrm{(Ai)}_-(z)J^{(\mathrm{Ai})}_k(z)$, $z\in\gamma_{k}$, $k=1,2,3$, where
\begin{equation*}
J^{(\mathrm{Ai})}_1(z)=\begin{pmatrix}1 & 1 \\ -1 & 0 \end{pmatrix},\
J^{(\mathrm{Ai})}_2(z)=\begin{pmatrix}1 & 1 \\ 0 & 1 \end{pmatrix},  \ J^{(\mathrm{Ai})}_3(z)=\begin{pmatrix}1 & 0 \\ -1 & 1 \end{pmatrix}.
\end{equation*}
\item[\rm (3)] 
As $z\to\infty$, we have
\begin{equation}\label{AiryAsyatinfty}
\Phi^{(\mathrm{Ai})}(z)=z^{-\frac{\sigma_{3}}{4} }\frac{1}{\sqrt{2}} \begin{pmatrix}i & 1\\ -i &1\end{pmatrix}\left(I+\frac{1}{48z^{\frac{3}{2}}}\begin{pmatrix}-1 & 6i \\ 6i & 1 \end{pmatrix}+O\left(z^{-3}\right)\right) e^{\frac{2}{3} z^{\frac{3}{2}} \sigma_{3}},
\end{equation}
where $\arg z\in(0,2\pi)$.
\end{rhp}

\subsection{Parabolic cylinder  parametrix}\label{PCP}
Define
$$Z_{0}(z)=2^{-\frac{\sigma_{3}}{2}}\begin{pmatrix}
D_{-\nu-1}(i z) & D_{\nu}(z) \\D_{-\nu-1}'(i z) & D_{\nu}'(z)\end{pmatrix}
\begin{pmatrix}
e^{i \frac{\pi}{2}(\nu+1)} & 0 \\
0 & 1
\end{pmatrix}
$$
and
$$Z_{n+1}(z)=Z_{n}(z) H_{n},\quad n=0,1,2,3,$$
where $D_{-\nu-1}$, $D_{\nu}$ are the standard parabolic cylinder functions (cf. \cite[Chapter 12]{Olver}),
$$
H_{0}=\begin{pmatrix}1 & 0 \\ h_{0} & 1\end{pmatrix}, \
H_{1}=\begin{pmatrix}1 & h_{1} \\ 0 & 1\end{pmatrix}, \
H_{n+2}=e^{i \pi\left(\nu+\frac{1}{2}\right) \sigma_{3}} H_{n} e^{-i \pi\left(\nu+\frac{1}{2}\right) \sigma_{3}}, \ n=0,1,2
$$
with
\begin{equation}\label{h0}
h_{0}=-i \frac{\sqrt{2 \pi}}{\Gamma(\nu+1)}, \quad h_{1}=\frac{\sqrt{2 \pi}}{\Gamma(-\nu)} e^{i \pi \nu}, \quad 1+h_{0} h_{1}=e^{2 \pi i \nu}.
\end{equation}
\begin{figure}[h]
  \centering
  \includegraphics[width=7cm,height=6.1cm]{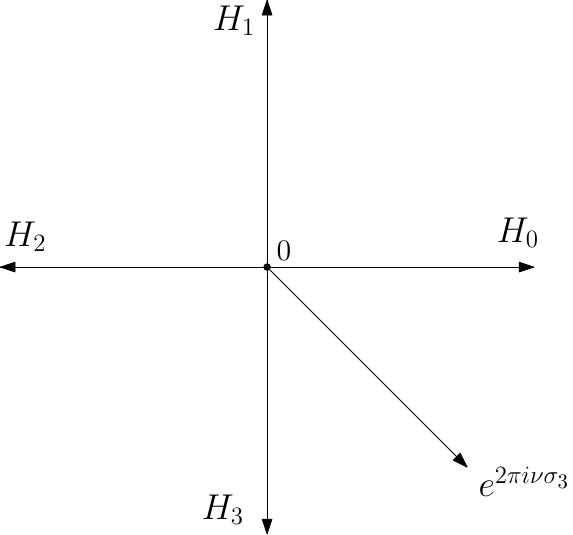}\\
  \caption{The jump contours and jump matrices for $\Phi^{(\mathrm{PC})}$}\label{PC}
\end{figure}

Let
$$
\Phi^{(\mathrm{PC})}(z)=\left\{\begin{aligned}
Z_{0}(z),\quad & \arg z \in\left(-\frac{\pi}{4}, 0\right), \\
Z_{1}(z),\quad & \arg z \in\left(0, \frac{\pi}{2}\right), \\
Z_{2}(z),\quad & \arg z \in\left(\frac{\pi}{2}, \pi\right), \\
Z_{3}(z),\quad & \arg z \in\left(\pi, \frac{3 \pi}{2}\right), \\
Z_{4}(z),\quad & \arg z \in\left(\frac{3 \pi}{2}, \frac{7 \pi}{4}\right),
\end{aligned}\right.
$$
then $\Phi^{(\mathrm{PC})}(z)$ solves the following RH problem (cf. \cite{BI,FIKN}).

\begin{rhp}
\item[\rm (1)] $\Phi^\mathrm{(PC)}(z)$ is analytic for all $z\in \mathbb{C}\setminus \bigcup^5_{k=1}\Sigma_{k}$, where $\Sigma_{k}=\{z\in\mathbb{C}:\arg z=\frac{k\pi}{2}\}$, $k=1,2,3,4$ and  $\Sigma_{5}=\{z\in\mathbb{C}:\arg z=-\frac{\pi}{4}\}$; see Figure \ref{PC}.
\item[\rm (2)] $\Phi^\mathrm{(PC)}(z)$ satisfies the jump conditions as shown in Figure \ref{PC}.
\item[\rm (3)] 
As $z\to\infty$, we have
\begin{align}\label{PCAsyatinfty}
\Phi^\mathrm{(PC)}(z)=\begin{pmatrix}0 &1 \\ 1 & -z\end{pmatrix}2^{\frac{\sigma_3}{2}}\begin{pmatrix}
1+O\left(\frac{1}{z^2}\right) & \frac{\nu}{z}+O\left(\frac{1}{z^3}\right) \\ \frac{1}{z}+O\left(\frac{1}{z^3}\right) & 1+O\left(\frac{1}{z^2}\right)\end{pmatrix}
e^{\left(\frac{z^{2}}{4}-\nu\log z\right) \sigma_{3}}.
\end{align}
\end{rhp}

\end{appendices}

\end{document}